\documentclass[a4paper,10pt]{amsart}
\usepackage{a4wide}
\usepackage{amssymb,amsmath,amsfonts} \allowdisplaybreaks
\usepackage{graphicx}
\usepackage{hyperref}

\usepackage{color}
\usepackage{amsthm}
\theoremstyle{plain}
\newtheorem{Theorem}{Theorem}[section]
\newtheorem{Proposition}[Theorem]{Proposition}
\newtheorem{Corollary}[Theorem]{Corollary}
\newtheorem{Lemma}[Theorem]{Lemma}
\theoremstyle{definition}
\newtheorem{Definition}[Theorem]{Definition}
\newtheorem{Example}[Theorem]{Example}
\newtheorem{assumption}[Theorem]{Assumption}

\theoremstyle{remark}
\newtheorem{Remark}[Theorem]{Remark}

\renewcommand{\epsilon}{\varepsilon}
\renewcommand{\phi}{\varphi}

\DeclareMathOperator{\Tr}{Tr}

\DeclareMathOperator{\diam}{diam}
\DeclareMathOperator{\supp}{supp}
\DeclareMathOperator{\rank}{rank}
\newcommand{\1}[0]{\textbf 1}
\newcommand{\EE}{\mathbb{E}}
\newcommand{\RR}{\mathbb{R}}
\newcommand{\CC}{\mathbb{C}}
\newcommand{\NN}{\mathbb{N}}
\newcommand{\ZZ}{\mathbb{Z}}

\newcommand{\hm}[1]{\textbf{*}\leavevmode{\marginpar{\tiny%
$\hbox to 0mm{\hspace*{-0.5mm}$\leftarrow$\hss}%
\vcenter{\vrule depth 0.1mm height 0.1mm width \the\marginparwidth}%
\hbox to 0mm{\hss$\rightarrow$\hspace*{-0.5mm}}$\\\relax\raggedright #1}}}
\title[A Banach space-valued ergodic theorem and the uniform approximation of the IDS]{A Banach space-valued ergodic theorem and the uniform approximation of the integrated density of states}
\author[D.~Lenz]{Daniel Lenz}
\address[D.L.]{Mathematisches Institut, Friedrich-Schiller-Universit\"at
Jena , 07743 Jena, Germany}
\urladdr{http://www.analysis-lenz.uni-jena.de/}

\author[F.~Schwarzenberger]{Fabian Schwarzenberger}
\address[F.S.]{Emmy-Noether-Projekt,
Fakult\"at f\"ur Mathematik, TU Chemnitz, 09107 Chemnitz, Germany}
\urladdr{http://www-user.tu-chemnitz.de/$\sim$fabis}

\author[I.~Veseli\'c]{Ivan Veseli\'c}
\address[I.V.]{Emmy-Noether-Projekt,
Fakult\"at f\"ur Mathematik, TU Chemnitz, 09107 Chemnitz, Germany}
\urladdr{http://www.tu-chemnitz.de/mathematik/enp}
\begin{document}
\begin{abstract}
In this paper we consider bounded operators on infinite graphs, in particular Cayley graphs of amenable groups. The operators satisfy an equivariance condition which is formulated in terms of a colouring of the vertex set of the underlying graph. In this setting it is natural to expect that the integrated density of states (IDS), or spectral distribution function, exists. We show that it can be defined as the uniform limit of approximants associated to finite matrices. The proof is based on a Banach space valued ergodic theorem which even allows  explicit convergence estimates. Our result applies to a variety of group structures and colouring types, in particular to periodic operators and percolation-type Hamiltonians.
\end{abstract}
\maketitle

\section{Introduction}

The topic of this paper is located at the interface between geometric group theory  and functional analysis.
The considered geometric setting is provided by a Cayley graph of an amenable group.
We study almost-additive functions mapping finite subsets of the graph into a fixed Banach space.
These functions are given in terms of a colouring of vertices of the graph.
In the applications we have in mind, the colouring is oftentimes a  realisation  of a stochastic field on the graph.
Under certain conditions on the tiling properties of the group,  and on frequencies of patterns occurring in the colouring,
we derive a Banach-space valued ergodic theorem.

This result is then applied to the study of the \emph{integrated density of states}
of equivariant Hamilton operators on graphs.
The integrated density of states  is also called \emph{spectral distribution function}. We abbreviate it in the sequel by IDS.
We show that the IDS can be approximated by normalised \emph{eigenvalue counting functions}
(associated to finite volume restrictions of the operator)
uniformly in the energy variable.

\medskip

Let us describe the contents of the paper more precisely.
At the end of this section we discuss earlier results which are related to ours.
In the next section we give  precise definitions of the graphs we are working on
and the conditions they should satisfy. In particular,
we require the existence of a \emph{F\o lner sequence} with the additional property that translates of each set
occurring in the sequence form a decomposition of the whole graph.
Moreover, we define the notions of a \emph{colouring} of a graph, of a \emph{pattern}
and a \emph{frequency} of a pattern along a F\o lner sequence.
Last but not least, in Section 3 it is explained
when a function taking finite subsets of the graph into a Banach space
is called \emph{almost additive} and \emph{invariant} with respect to a colouring.

For such functions we formulate and prove
a \emph{Banach-space valued ergodic theorem} in Section 3.
The subsequent section is devoted to \emph{Hamilton operators
on graphs} and their IDS. The link to the previous sections is given by the
requirement that the Hamilton operator is in an appropriate sense equivariant with
respect to the colouring of a graph. For each finite subset of
the graph there is a natural restriction of the operator to this subset
which is given  by a finite dimensional matrix.
The eigenvalue counting function of such a matrix
is an element of the Banach space of bounded, right-continuous functions. The map
taking finite subsets of the graph into the mentioned Banach space
can be shown  to be almost additive and invariant. Thus the
ergodic theorem of Section 3 can be applied and yields the
uniform approximability of the IDS by its finite volume analogues.

In the penultimate section we discuss specific examples where the
assumptions of our theorems are satisfied.
On the one hand we give for the euclidean lattice $\ZZ^d$
and the Heisenberg group $H_3$ explicit F\o lner sequences which satisfy the geometric requirements
needed to deal with almost additive functions. On the other
hand we show that Laplace operators on periodic graphs and percolation subgraphs thereof
define a Hamiltonian which is invariant with respect to a colouring.
More precisely, up to a set of measure zero, each
realisation of the percolation process gives a colouring with
well-defined frequencies of patterns. The Laplacian on the percolation
subgraph is invariant with respect to this colouring. The periodic model corresponds to
the trivial case where the percolation cluster is the whole graph.

In the last section we discuss some open questions.
\medskip

Let us compare the results of this paper to previous ones.
The approximability or convergence of the IDS is a
topic well studied in the mathematical physics literature. For ergodic
random and almost-periodic operators in Euclidean space the existence of the
IDS was first rigorously proven in the seminal papers \cite{Pastur-71, Shubin-79}.
See \cite{KirschM-07,Veselic-07b} for surveys on this topic and
a discussion of the extensive literature up to '07.
In the mathematical physics literature one typically considers operators on euclidean spaces
$\RR^d$ or $\ZZ^d$ which are invariant and ergodic with respect
to a (commutative)  group of translations acting on the underlying space.
In such a setting, it is well known that the IDS converges
almost surely and at all energies at which the IDS is
continuous. Let us note that this type of convergence is weaker than
pointwise convergence at all energies, let alone uniform convergence.

In the geometry literature the IDS is called spectral
distribution function. It has been studied for Lapacians and periodic
Schr\"odinger operators on covering manifolds and periodic graphs.
In various instances the approximability of the IDS by its finite
volume analogues has been established in the topology of pointwise convergence.
For periodic and random Schr\"odinger operators on manifolds this has been done in
\cite{Sznitman-89,Sznitman-90,AdachiS-93,PeyerimhoffV-02,LenzPV-04} and for finite difference operators
on periodic graphs in \cite{MathaiY-02,MathaiSY-03,DodziukLMSY-03,Veselic-05b}.
Previously, the approximability of $\ell^2$-Betti numbers has been established
in \cite{DodziukM-97,DodziukM-98,Eckmann-99,LueckS-99}. Since the zeroth $\ell^2$-Betti number is the value of
the IDS at energy zero, these results are intimately connected. They can be understood as dealing with existence of the IDS in a single point.  The monograph
\cite{Lueck-02} and the collection of talks \cite{DodziukLSV-06} give a survey
of results in this direction and an overview of the literature up to '06.

Having discussed pointwise convergence of the normalised eigenvalue counting functions
to the IDS, we now turn to  previous results which concern the convergence  \emph{with respect to the supremum norm}. Here, we discuss earlier papers in some detail, since they are  closely related to our own result. There are essentially two approaches to uniform convergence results:

One approach is based on Banach-space valued ergodic theorems.
The uniform convergence of the IDS is then nothing but a special instance of a geometry based averaging result valid for a rather general class of functions.
This approach has been first developed for   quasi-periodic tilings  in Euclidean space  in \cite{LenzS-06}
(see \cite{GeerseH-91} and \cite{Besbes-08} for  related ergodic theorems as well).
A similar ergodic type result was then established for colourings of $\ZZ^d$ (see \cite{Lenz-02} for the one-dimensional case as well)  and  used to prove uniform convergence of the
IDS for corresponding models  in  \cite{LenzMV-08}.
These papers are strongly rooted in the rather simple geometry of Euclidean spaces and  their sublattices.
On the conceptual level, our main goal here is to achieve  an extension of the approach in \cite{LenzMV-08} to the non abelian situation.

As far as uniform convergence of the IDS is concerned there exists also  a different approach. This approach  does not use (or give)  a Banach space valued averaging procedure for general functions. It is  rather tailored to deal with the IDS  as it relies on special properties of converging sequences of measures on the real line.  This is first discussed in the
preprint \cite{Elek-b}
for a class of graphs called \emph{abstract quasi-crystal graphs}, see also \cite{Elek-a}.
In \cite{LenzV-09} uniform convergence of the IDS was  then established for a
large class of equivariant ergodic Hamiltonians on graphs.
While the geometric setting there includes our main examples discussed in the penultimate section it falls short of our  present results in at least two ways:
It does not give the explicit error bounds we achieve here and it can not be used to treat single operators but can only deal with ergodic families.

To illustrate the relevance of the two last mentioned features let us discuss in some more
detail the implications of the approximability of the IDS.
The main theorem in Section 4 establishes that the IDS is well
defined. More precisely, it is shown that the sequence of
distribution functions (the normalised eigenvalue counting
functions) converges to some element of a Banach space.
In certain situations there is an alternative way to define the
IDS in operator-algebraic terms. The corresponding formula is
sometimes called Pastur-Shubin trace formula, cf.~Remark
\ref{remark_pastur}.
Thus the question arises whether the two definitions of the
IDS coincide. Indeed, the answer to this question depends on the amenability properties of the underlying space. There are examples where the two definitions lead to two different distribution functions, see e.\,g.\, Section 4 of \cite{AdachiS-93}.

In the context of mathematical physics the definition using the approximation by eigenvalue counting functions is more relevant. It corresponds more closely to the intuition of physicists. In fact, properties of the infinite volume system are oftentimes studied via its finite volume approximants, and vice versa. Let us point out several examples for this approach.

 An key feature of the IDS is its asymptotics near the bottom of the spectrum of the underlying operator.
If the latter is periodic one expects a van Hove behaviour, i.e. polynomial vanishing of the IDS.
If the operator is `sufficiently random', one expects a Lifshitz behaviour, namely that the IDS goes to zero at an exponential rate.
To establish this result one estimates the lowest eigenvalue of a properly scaled finite volume system.
However, this bound is only useful if one has control of the difference between the IDS and its finite volume approximant.
For this purpose various approaches have been implemented in the euclidan setting.
 For instance one can sandwich the IDS between two different finite volume approximants,
one giving a lower and the other an upper bound, and control their distance. The papers
\cite {KirschM-83a} and \cite{Simon-85b} are instances where this method has been implemented.
A different approach is to use the specific structure of the operator under consideration
to obtain very fast convergence of the \emph{averaged} finite volume IDS to its infinite volume analogue, cf.~e.g.{} \cite{Klopp-99}.
More information on this topic can be found in the overview article \cite{KirschM-07} and the references therein.

From the physics point of view it is also important to understand the quantitative continuity and differentiability properties of the IDS.
This is often achieved by first proving a related statement (called Wegner estimate) for finite volume systems and then carrying it over to the infinite volume. In special situations it is expedient to go in the reverse direction, see for instance in \cite{CarmonaKM-87}. For detailed information see e.\,g.\, the surveys \cite{KirschM-07,Veselic-07b}.

We finish this section by shortly discussing the relevance of our single colouring  based approach compared to an approach dealing with ergodic families of colourings. Basically, there are two advantages. One advantage is the  explicit description of the set of colourings where convergence holds. More precisely, convergence holds whenever frequencies exist. The other advantage is that we can deal with situations in which there is no natural ergodic system at ones disposal. A simple example of such a system is given by the set of  visible points in $\ZZ^d$.  This set consists of those  points in $\ZZ^d$ which can be seen from the origin i.e. whose coordinates have greatest common divisor equal to one. By colouring the visible points with one colour and the other points of $\ZZ^d$ with another colour one obtains a  colouring.
In this example frequencies exist along cubes centered at the origin. However, there does not seem to be a natural ergodic system. The corresponding ergodic theorem in this  (abelian) situation was already discussed in \cite{LenzMV-08} to which we also refer for further references on the set of visible points.


\section{Basics and notation}
 This paper concerns operators on Cayley graphs of finitely generated, amenable groups. The group will usually be denoted by $G$, a finite generating system by $S$ and the unit element by $e$. We assume $S$ to be symmetric, i.e. $s\in S \Rightarrow s^{-1}\in S$. A word with letters in $S$ is a finite sequence $w=(s_1,\dots,s_L)$ of elements from $S$. The length of such a word $w$ is $L$, the number of letters, and the value $\overline w$ of $w$ is the element of $G$ one gets by evaluating the product $\overline w=s_1\cdots s_L$. We define the distance between two given elements $g,h\in G$ as the minimal length of a word with value $g^{-1}h$, which gives us the so called \textit{word metric} $d_S: G\times G\rightarrow \mathbb N_0$
\[
d_S(g,h):=\min\{L\in \NN_0\ \vert \ s_0=e,\ \exists s_1,\dots,s_L\in S \mbox{ such that } s_0 s_1\cdots s_L = g^{-1}h\}.
\] We denote a ball of radius $R$ around an element $x\in G$ by $B_R(x):=\{g\in G\vert d_S(g,x)\leq R\}$ and the ball around the unit element by $B_R:=B_R(e)$. Given a finite subset $Q\subseteq G$ we define the \textit{diameter} by $\diam (Q):=\max \{d_S(g,h)\vert g,h\in Q\}$ and use $\vert Q\vert$ to denote the \textit{cardinality} of $Q$. Furthermore we introduce the following notation related to the boundary of a finite subset $Q\subseteq G$:
\begin{equation}
\begin{split}
 &\partial_{int}^R(Q):= \{x\in Q\,\vert\, d_S(x,G\setminus Q)\leq R\}\hspace{1cm} \partial_{ext}^R(Q) := \{x\in G\setminus Q\,\vert\, d_S(x,Q)\leq R\} \\[1ex]
&\partial^R(Q)      := \partial_{int}^R(Q)\cup\partial_{ext}^R(Q) \hspace{1cm} Q_R:=  Q\setminus \partial^R (Q)\hspace{1cm} Q^R:=Q\cup \partial^R(Q)
\end{split}
\end{equation}
We use the notations $(Q_n)$ and $(Q_n)_{n\in \NN}$ for a sequence of finite subsets of $G$, where the index $n$ takes values in $\NN$ and we write  $Q_{n,R}$ instead of $(Q_n)_R$. It is well known \cite{Adachi-93}, that \textit{amenability} of $G$ is equivalent to the existence of a sequence $(Q_n)_{n\in \mathbb N}$ of finite subsets of $G$ such that
\begin{equation*}
\lim\limits_{n\rightarrow \infty} \frac{\vert Q_n S\setminus Q_n\vert}{\vert Q_n\vert}=0
\end{equation*}
holds. Such a sequence $(Q_n)_{n\in \mathbb N}$ is called \textit{F\o{}lner sequence}. Given a set $Q\subseteq G$ a \textit{partition} of $Q$ is a set of pairwise disjoint subsets $Q_i$, $i\in I$ of $Q$ such that $\bigcup_{i\in I} Q_i=Q$, where $I$ is some index set. We say that $Q\subseteq G$ \textit{tiles} the group $G$ or $Q$ is a \textit{monotile} of $G$ if there exists a set $K\subseteq G$ such that $\{Qg\ \vert\ g\in K\}$ is a partition of $G$. In this case $\{Qg\ \vert\ g\in K\}$ is called a \textit{tiling} of the group along the \emph{grid} $K$.
If additionally $K=K^{-1}$ holds, we say that $Q$ \emph{symmetrically tiles} $G$ or $\{Qg\ \vert\ g\in K\}$ is a \emph{symmetric tiling} of $G$.
An assumption on the group $G$ will be the following:
there exists a F\o lner sequence $(Q_n)$ such that for each $n\in\NN$ the set $Q_n$ symmetrically tiles $G$.

\medskip
In order to define colourings and patterns we denote the set of all finite subsets of $G$ by $\mathcal{F}(G)$ and introduce an arbitrary finite set ${\mathcal A}$, which could be seen as the set of possible colours. A \textit{colouring} is a map ${\mathcal C}: G\rightarrow \mathcal A$ and a \textit{pattern} is a map $P: D(P) \rightarrow \mathcal A$, where $D(P)\in {\mathcal F}(G)$ is called the domain of $P$. The \textit{set of all patterns} is denoted by ${\mathcal P}$ and for a fixed $Q\in {\mathcal F}(G)$ the subset of ${\mathcal P}$ which contains only the patterns with domain $Q$ is denoted by ${\mathcal P}(Q)$. Given a set $Q\subseteq D(P)$ and an element $x\in G$ we furthermore define a \textit{restriction of a pattern} by $P\vert_Q : Q\rightarrow {\mathcal A}$, $g\mapsto P\vert_Q(g)=P(g)$ and a \textit{translation of a pattern} $Px: D(P)x\rightarrow {\mathcal A}, yx\mapsto P(y)$. Two patterns are called \textit{equivalent} if one is a translation of the other. The equivalence class of a pattern $P$ is denoted by $\tilde P$. We write $\tilde {\mathcal P} $ for the induced set of equivalence classes in $\mathcal P$.
For two patterns $P$ and $P'$ the number of occurrences of the pattern $P$ in $P'$ is denoted by
\[\sharp_P(P'):=\left|\{x\in G\,\vert\, D(P)x\subseteq D(P'), P'\vert_{D(P)x}= Px\}\right|.\]
Counting occurrences of patterns along a F\o{}lner sequence $(U_j)_{j\in \mathbb N}$ leads to the definition of frequencies. If for a pattern $P$ and a F\o{}lner sequence $(U_j)_{j\in \mathbb N}$ the limit
\[\nu_P:=\lim\limits_{j\rightarrow \infty}\frac{\sharp_P({\mathcal C}\vert_{ U_j})}{\vert U_j \vert}.\]
exists, we call $\nu_P$ the \textit{frequency of $P$ in the colouring ${\mathcal C}$ along $(U_j)_{j\in \mathbb N}$}.

\begin{Lemma}
\label{lemma_abc}
 \begin{itemize}
  \item [(a)] For an arbitrary $Q\in {\mathcal F}(G)$ one has the equalities
   \[\partial_{int}^R(Q)=\bigcup\limits_{s\in B_R}(Q\setminus  Qs)\hspace{1cm}\mbox{ and }\hspace{1cm}\partial_{ext}^R(Q)=\bigcup\limits_{s\in B_R}(Qs\setminus  Q).\]
  \item [(b)] If $(U_j)_{j\in \NN}$ is a F\o{}lner sequence we have for any $R\in \NN$
  \[\lim\limits_{j\rightarrow \infty} \frac{\vert \partial_{int}^R U_j \vert}{\vert U_j\vert}=0,\hspace{0.7cm}
  \lim\limits_{j\rightarrow \infty} \frac{\vert \partial_{ext}^R U_j \vert}{\vert U_j\vert}=0  \hspace{0.7cm}\mbox{ and }\hspace{0.7cm}
  \lim\limits_{j\rightarrow \infty} \frac{\vert \partial^R U_j \vert}{\vert U_j\vert}=0.\]
  \item [(c)] The sequence $(U_{j,R})_{j\in \NN}$ is a F\o{}lner sequence if $(U_j)_{j\in \NN}$ is one.
  \item [(d)] If $\nu_P$ is the frequency of a pattern $P$ along the F\o{}lner sequence $(U_j)_{j\in \mathbb N}$ then $\nu_P$ is the frequency of $P$ along $(U_{j,R})_{j\in \NN}$ as well.
 \end{itemize}
\end{Lemma}
The proof of this Lemma is stated in the appendix.

\begin{Remark}
 Certain considerations outlined in Lemma \ref{lemma_abc} can be simplified if the identity element $e$ is contained in the generating set $S$. In this situation we have that the sequence $(U_j)_{j\in\NN}$ is a F\o{}lner sequence if and only if
  \begin{equation}\label{lemma_abc1}
    \lim_{j\rightarrow\infty }\frac{\vert U_j S\vert}{\vert U_j\vert}=1.
  \end{equation}
 For the proof of this note that $U_j\subseteq U_jS$ since $S$ contains the unit element $e$. Therefore we have
\[
 \frac{\vert U_{j}S\setminus U_{j}\vert}{\vert U_{j} \vert}
= \frac{\vert U_{j}S\vert - \vert U_{j}\vert}{\vert U_{j} \vert}
= \frac{\vert U_{j}S\vert}{\vert U_{j} \vert} -1 .
\]
which directly implies the claimed equivalence.
\end{Remark}

In order to state the announced ergodic theorem for Banach valued functions, we need the following definitions:

\begin{Definition}\label{def_boundaryterm}
A function $b :{\mathcal F}(G)\rightarrow [0,\infty)$ is called a \textnormal{boundary term} if
 \begin{itemize}
   \item [(a)] $b(Q)=b(Qx)$ for all $x\in G$ and all $Q\in {\mathcal F}(G)$,
   \item [(b)] $\lim\limits_{j\rightarrow\infty} \frac{b(U_j)}{\vert U_j \vert}=0$ for any F\o{}lner sequence $(U_j)$ and
   \item [(c)] there exists $D>0$ with $b(Q)\leq D\vert Q\vert$ for all $Q\in {\mathcal F}(G)$.
 \end{itemize}
\end{Definition}
For a pattern $P$ we define $b(P):=b(D(P))$. Due to property (a) the value $b(P)$ depends only on the equivalence class of a pattern. Thus $b(\tilde P):=b(P)$ is well-defined.
\begin{Definition}\label{def_Fschlange}
Let $(X,\Vert\cdot\Vert)$ be a Banach space and $\tilde F$ a function $\tilde F: \tilde{\mathcal P}\rightarrow X$. We call $\tilde F$
\textnormal{almost-additive} if there exists a boundary term $b$ such that for any $\tilde P\in \tilde {\mathcal P}$ and any disjoint decomposition $P=\bigcup_{k=1}^m P_k$ of a representative $P$ of $\tilde P$ we have:
\[\left\Vert \tilde F(\tilde{P})-\sum\limits_{k=1}^m \tilde F(\tilde P_k)\right\Vert \leq \sum\limits_{k=1}^m b(\tilde P_k).\]
\end{Definition}
Let $\tilde F: \tilde{\mathcal P}\rightarrow X$ be an almost additive function and $P$ an arbitrary pattern. If we decompose $P$ into $\vert D(P)\vert$ one-element patterns $P_k$ with disjoint domains, i.e. $P=\bigcup P_k$ we obtain
\begin{equation}\label{ftildebounded0}
\|\tilde F(\tilde P) \|
 \le \|\tilde F(\tilde P)-\sum_k \tilde F( \tilde P_k) \| + \| \sum_k \tilde F(  \tilde P_k) \|
 \le \sum_k b(\tilde P_k) + \sum_k \| \tilde F(\tilde P_k) \|.
\end{equation}
Since all $D(P_k)$ contain exactly one element, $b(\tilde P_k)$ is independent of $k$ and $F(\tilde P_k)$ can take at most $\vert{\mathcal A}\vert$ different values. Therefore it follows
\begin{equation}\label{Ftildebounded}
\|\tilde F(\tilde P) \|
  \le C|D(P)|\hspace{1cm}\mbox{with } C= b(e)   + \max_{a \in{\mathcal A}} \| \tilde F(\tilde {e_a}) \|.
\end{equation}
Here $e$ is the unit element considered as a vertex of the Cayley graph, and  $e_a: \{e\}\rightarrow {\mathcal A}$ the map $e_a(e)=a$. A given colouring ${\mathcal C}$ on $G$ and an almost-additive function $\tilde F: \tilde{\mathcal P}\rightarrow X$ give rise to a function
\[F: {\mathcal F}(G)\rightarrow X, \hspace{0.7cm} F(Q):= \tilde F(\tilde{{\mathcal C}\vert_{Q}})\hspace{1cm}\mbox{for } Q\in {\mathcal F}(G).\]
The following properties of $F$ obviously hold
\begin{itemize}
 \item [(i)] \emph{${\mathcal C}$-invariant}: if $x\in G$ is such that the patterns ${\mathcal C}\vert_Q$ and ${\mathcal C}\vert_{Qx}$ are equivalent, then we have \[ F(Q)=F(Qx),\]
\item [(ii)] \emph{almost additive}: if $Q_k, k=1,\dots,n$ are disjoint subsets of $G$, then we have
\[\bigg\Vert F\Big(\bigcup\limits_{k=1}^m Q_k\Big)-\sum\limits_{k=1}^m F(Q_k)\bigg\Vert \leq \sum\limits_{k=1}^m b(Q_k),\]
\item [(iii)] \emph{bounded}: there exists a $C > 0$ such that
\[\Vert F(Q)\Vert\leq C\vert Q\vert \hspace{1cm}\mbox{for all } Q\in {\mathcal F}(G).\]
\end{itemize}
A $\mathcal C$-invariant and almost additive function $F: {\mathcal F}(G)\rightarrow X$ is automatically bounded. This follows from an estimate analogous to (\ref{ftildebounded0}). Instead of defining $F$ based on $\tilde F$ one could also proceed the other way around: if a function $F:{\mathcal F}(G)\rightarrow X$ with the properties (i) and (ii) is given, define $\tilde F: \tilde{\mathcal P}\rightarrow X$ by the following procedure. If for $\tilde P\in \tilde{\mathcal P}$ there exists an $Q\in {\mathcal F}(G)$ such that $\tilde {{\mathcal C}\vert_Q}=\tilde P$ set $\tilde F(\tilde P)=F(Q)$. This definition is independent of the particular choice of $Q$ by the $\mathcal C$-invariance of $F$. If such a $Q$ does not exist, set $F(\tilde P)=0$.
Therefore showing (i) and (ii) for $F$ is the same as showing almost additivity for $\tilde F$. To simplify the notation we will write $\tilde F(P)$ instead of $\tilde F(\tilde P)$ for a given pattern $P\in {\mathcal P}$.
In order to be able to refer to it later, we summarise the assumptions made in this section:

\begin{assumption}\label{ass:first}
 $G$ is an amenable group generated by a finite and symmetric set $S$, ${\mathcal A}$ is a finite set and ${\mathcal C}: G\rightarrow {\mathcal A}$ is a map, which we will call a colouring.
There exists a F\o lner sequence $(Q_n)_{n\in \NN}$ in $G$ such that each element of this sequence symmetrically tiles the group.
$(U_j)_{j\in \mathbb N}$ is a F\o{}lner sequence along which the frequencies $\nu_P=\lim_{j\rightarrow \infty}\vert U_j\vert^{-1} \sharp_P({\mathcal C}\vert_{U_j})$ exist for all patterns $P\in \bigcup_{n\in \NN}{\mathcal P}(Q_n)$. $(X,\Vert \cdot \Vert)$ is a Banach-space.
\end{assumption}

\begin{Remark}\label{rem_weiss}
 Let us discuss Assumption \ref{ass:first}.
\begin{itemize}
\item We assume that $G$ contains a F\o lner sequence $(Q_n)$ such that each $Q_n$ symmetrically tiles $G$.
This condition is particularly satisfied if there exists a sequence of subgroups $(G_n)_{n\in\NN}$ such that one can choose the associated fundamental domains $(Q_n)_{n\in \mathbb N}$ to be a F\o{}lner sequence. Based on a result of Weiss \cite{Weiss-01}, Krieger proves in \cite{Krieger-07} that this is fullfilled for any residually finite, amenable group. This gives that in particular any group of polynomial volume growth fits in our framework.
\item In the special case where the group equals $\ZZ^d$ it is convenient to think of the sets $U_j$ as balls of radius $j$ and of $Q_n$ as cubes of side length $n$. While both of them are F\o lner sequences, $(Q_n)$ has the additional property that each $Q_n$ symmetrically tiles $\ZZ^d$. Here we want the frequencies of the patterns to exist along the sequence of balls.
 \end{itemize}

\end{Remark}


\section{An ergodic theorem}
Given the setting outlined in the previous section we are in the position to formulate the announced ergodic type theorem for certain Banach-space valued functions.

\begin{Theorem}\label{ergthm}
 Assume \ref{ass:first}. For a given ${\mathcal C}$-invariant and almost-additive function $F: {\mathcal F}(G)\rightarrow X$ the following limits
\[\overline F := \lim\limits_{j\rightarrow \infty}\frac{F(U_j)}{\vert U_j \vert}= \lim\limits_{n \rightarrow\infty} \sum_{P\in {\mathcal P}(Q_n)} \nu_P \frac{\tilde F(P)}{\vert Q_n\vert}  \]
exist and are equal. Furthermore, for $j,n\in \mathbb N$ the difference
\[\Delta(j,n):=\Big\Vert \frac{F(U_j)}{\vert U_j \vert} -   \sum_{P\in {\mathcal P}(Q_n)} \nu_P \frac{\tilde F(P)}{\vert Q_n\vert}\Big\Vert\]
satisfies the estimate
\begin{equation}\label{estimate}
\Delta(j,n)\leq \frac{b(Q_n)}{\vert Q_n\vert} + (C+D)\frac{\vert \partial^{\diam(Q_n)}U_j  \vert}{\vert U_j\vert}+C\sum\limits_{P\in {\mathcal P}(Q_n)} \Big \vert \frac{\sharp_P ({\mathcal C}\vert_{U_j})}{\vert U_j\vert }- \nu_P  \Big \vert.
\end{equation}
\end{Theorem}
\begin{proof}
We firstly prove (\ref{estimate}). By adding a zero we get
\begin{eqnarray*}
 \Delta(j,n)&=&\Big\Vert \frac{F(U_j)}{\vert U_j \vert}- \sum\limits_{P\in {\mathcal P}(Q_n)}\frac{\sharp_P({\mathcal C}\vert_{U_j})}{\vert U_j\vert}\frac{\tilde F(P)}{\vert Q_n\vert} + \sum\limits_{P\in {\mathcal P}(Q_n)}\frac{\sharp_P({\mathcal C}\vert_{U_j})}{\vert U_j\vert}\frac{\tilde F(P)}{\vert Q_n\vert}-   \sum_{P\in {\mathcal P}(Q_n)} \nu_P \frac{\tilde F(P)}{\vert Q_n\vert}\Big\Vert \\
&\leq&\Big\Vert \frac{F(U_j)}{\vert U_j \vert}- \sum\limits_{P\in {\mathcal P}(Q_n)}\frac{\sharp_P({\mathcal C}\vert_{U_j})}{\vert U_j\vert}\frac{\tilde F(P)}{\vert Q_n\vert}\Big\Vert  + \Big\Vert \sum\limits_{P\in {\mathcal P}(Q_n)}\Big(\frac{\sharp_P({\mathcal C}\vert_{U_j})}{\vert U_j\vert}-   \nu_P \Big )\frac{\tilde F(P)}{\vert Q_n\vert}\Big\Vert .
\end{eqnarray*}
With another application of the triangle inequality this gives
\[\Delta(j,n)\leq D_1(j,n)+D_2(j,n),\]
where
\begin{eqnarray*}
 D_1(j,n)&:=& \Big\Vert \frac{F(U_j)}{\vert U_j \vert}- \sum\limits_{P\in {\mathcal P}(Q_n)}\frac{\sharp_P({\mathcal C}\vert_{U_j})}{\vert U_j\vert}\frac{\tilde F(P)}{\vert Q_n\vert}\Big\Vert \\
 D_2(j,n)&:=&\sum\limits_{P\in {\mathcal P}(Q_n)}\Big\vert\frac{\sharp_P({\mathcal C}\vert_{U_j})}{\vert U_j\vert}-   \nu_P \Big \vert\frac{\Vert\tilde F(P)\Vert}{\vert Q_n\vert}.
\end{eqnarray*}
We use the boundedness of $\tilde F$, see (\ref{Ftildebounded})
\begin{equation}
\Vert \tilde F(P) \Vert \leq  C\vert Q_n \vert,
\end{equation}
to obtain
\begin{equation}\label{bnd_d2}
D_2(j,n)\leq C \sum_{P\in {\mathcal P}(Q_n)} \Big \vert \frac{\sharp_P({\mathcal C}\vert_{U_j})}{\vert U_j\vert}-\nu_P\Big \vert.
\end{equation}

For each fixed $n\in \NN$ the set $Q_n$ (symmetrically) tiles the group $G$, i.e. there exists a symmetric set $G_n\subseteq G$ such that $G=\bigcup_{g\in G_n}Q_n g$, where $Q_n g\cap Q_n h=\emptyset$ for all $g,h\in G_n$ with $g\neq h$.
This property remains valid after shifting the grid by an arbitrary $x^{-1}\in G$.
In fact we have  
\[G=Gx^{-1}=\bigcup\limits_{g\in G_n} Q_n gx^{-1}=\bigcup\limits_{g\in G_n x^{-1}} Q_n g,\]
where $G_n x^{-1}:=\{gx^{-1}\vert g\in G_n\}$. Still $Q_n g\cap Q_n h=\emptyset$ holds for all distinct $g,h\in G_n x^{-1}$, since $g=\tilde gx^{-1}$ and $h=\tilde hx^{-1}$ for some distinct $\tilde g,\tilde h\in G_n$ and
\[
 Q_n g\cap Q_n h=\emptyset
 \quad \Leftrightarrow \quad Q_n \tilde gx^{-1}\cap Q_n \tilde hx^{-1}=\emptyset
 \quad \Leftrightarrow \quad Q_n \tilde g\cap Q_n \tilde h=\emptyset.
\]
 Given a set $A\in {\mathcal F}(G)$ and an element $x\in G$, we introduce the set of elements $g\in G_n x^{-1}$ which gives rise to a translate $Q_n g$, which is not disjoint from $A$:
\[S(A,x,n):=\{g \in G_n x^{-1}\,\vert\, Q_ng \cap A \neq \emptyset\}.\]
We distinguish two types of elements in $S(A,x,n)$
\[I(A,x,n):=\{g \in G_n x^{-1}\,\vert\, Q_ng \subseteq A \} \hspace{0.6cm}\mbox{and}\hspace{0.6cm}\partial (A,x,n):=S(A,x,n) \setminus I(A,x,n).\]
Since we have $Q_ng\subseteq \partial^{\diam Q_n} A$ for all $g\in \partial(A,x,n)$ and $Q_ng\subseteq A$ for all $g\in I(A,x,n)$ the disjointness of the translates implies that the following inequalities hold:
\begin{equation}\label{diamest}
\vert \partial (A,x,n) \vert \cdot \vert Q_n\vert \leq \vert \partial^{\diam Q_n} A\vert          \hspace{1cm}\mbox{and}\hspace{1cm}      \vert I (A,x,n) \vert \cdot \vert Q_n\vert \leq \vert A \vert.
\end{equation}
Given an $n\in \mathbb N$, $A\in {\mathcal F}(G)$ and $x\in G$ we have $Q_n g=Q_n g\cap A$ for $g\in I(A,x,n)$ and thus
\begin{eqnarray*}
 T(A,x,n)&:=& \bigg\Vert F(A) - \sum\limits_{g\in I(A,x,n)} F(Q_n g)\bigg\Vert = \bigg\Vert F(A) - \sum\limits_{g\in I(A,x,n)} F(Q_n g\cap A)\bigg\Vert  \\
&\leq& \bigg\Vert F(A) -  \sum\limits_{g\in S(A,x,n)} F(Q_n g\cap A) \bigg\Vert + \bigg\Vert \sum\limits_{g\in \partial(A,x,n)} F(Q_n g \cap A)   \bigg \Vert
\end{eqnarray*}
where the last inequality holds since $S(A,x,n)$ is the disjoint union of $\partial(A,x,n)$ and $I(A,x,n)$. Now we use the almost additivity and the boundedness of $F$ and later on the properties of the boundary term $b$ to obtain
\begin{eqnarray*}
 T(A,x,n)&\leq& \bigg ( \sum\limits_{g\in I(A,x,n)}b(Q_n g) +\sum\limits_{g\in \partial(A,x,n)}b(Q_n g \cap A)\bigg )+ \sum\limits_{g\in \partial(A,x,n)}C\vert Q_n g \vert\\
&\leq& \sum\limits_{g\in I(A,x,n)}b(Q_n ) +\sum\limits_{g\in \partial(A,x,n)}D \vert Q_n  \vert+ \sum\limits_{g\in \partial(A,x,n)}C\vert Q_n  \vert \\
&\leq& \vert I(A,x,n)\vert b(Q_n) + \vert \partial(A,x,n)\vert (C+D)\vert Q_n\vert.
\end{eqnarray*}
The inequalities (\ref{diamest}) yield the estimate
\begin{equation}\label{est_T}
T(A,x,n)\leq \frac{\vert A\vert}{\vert Q_n\vert}b(Q_n) + (C+D)\vert \partial^{\diam Q_n} A\vert .
\end{equation}
Furthermore we have the equality
\begin{equation}\label{eqsets}
 \{z\in G\,\vert\, Q_n z \subseteq A\} = \dot {\bigcup_{x\in Q_n}}\{z\in G_nx^{-1}\,\vert\, Q_nz\subseteq A\}
\end{equation}
since for each $z\in G$ there is $x\in Q_n$ and $g\in G_n$ with $z^{-1}=xg$ and hence $z= g^{-1}x^{-1}\in G_n x^{-1}$, as $G_n$ is a symmetric subset of $G$. To see that the union in \eqref{eqsets} is disjoint observe that for given $x,y\in Q_n$ with $x\neq y$ and $z\in G_nx^{-1}$ we have $z^{-1}\in xG_n$. By the tiling property of $Q_n$ this gives $z^{-1}\notin yG_n$ and hence $z\notin G_ny^{-1}$.

The $\mathcal C$-invariance of $F$ and the equation \eqref{eqsets} imply
\begin{equation}\label{FtotildeF}
 \sum\limits_{P\in {\mathcal P}(Q_n)}\sharp_P({\mathcal C}\vert_{U_j})\tilde F(P)=\sum\limits_{z\in G : Q_nz\subseteq U_j}F(Q_nz)=\sum\limits_{x\in Q_n} \sum\limits_{g\in I(U_j,x,n) }F(Q_n g),
\end{equation}
from which we deduce
\[\vert U_j\vert D_1(j,n)= \Big\Vert F(U_j)- \sum\limits_{P\in {\mathcal P}(Q_n)}\sharp_P({\mathcal C}\vert_{U_j}) \frac{\tilde F(P)}{\vert Q_n\vert}\Big\Vert \stackrel{(\ref{FtotildeF})}{=} \Big\Vert F(U_j)- \frac{1}{\vert Q_n\vert }  \sum\limits_{x\in Q_n} \sum\limits_{g\in I(U_j,x,n) }F(Q_n g)  \Big\Vert. \]
Using $\frac{1}{\vert Q_n\vert}\sum_{x\in Q_n}1=1$ we get
\begin{eqnarray*}
\Big\Vert F(U_j)- \frac{1}{\vert Q_n\vert }  \sum\limits_{x\in Q_n} \sum\limits_{g\in I(U_j,x,n) }F(Q_n g)  \Big\Vert&=&\frac{1}{\vert Q_n \vert} \Big\Vert  \sum\limits_{x\in Q_n} \Big( F(U_j)-  \sum\limits_{g\in I(U_j,x,n) }F(Q_n g)\Big)  \Big\Vert \\
&\leq& \frac{1}{\vert Q_n \vert}  \sum\limits_{x\in Q_n} \underbrace{\Big\Vert  F(U_j)-  \sum\limits_{g\in I(U_j,x,n) }F(Q_n g)\Big\Vert}_{T(U_j,x,n)} .
\end{eqnarray*}
Now we use the estimate (\ref{est_T}) for $T(U_j,x,n)$ to obtain

\begin{eqnarray*}
 D_1(j,n)&\leq&\frac{1}{\vert Q_n\vert}  \sum\limits_{x\in Q_n} \left(\frac{b(Q_n)}{\vert Q_n\vert} + (C+D)\frac{\vert \partial^{\diam Q_n} U_j\vert}{\vert U_j\vert}\right)\\
&=& \frac{b(Q_n)}{\vert Q_n\vert} +  (C+D)\frac{\vert \partial^{\diam Q_n} U_j\vert}{\vert U_j\vert}.
\end{eqnarray*}
Together with the upper bound for $D_2(j,n)$ in (\ref{bnd_d2}) we have
\begin{eqnarray*}
  \Delta(j,n)&\leq &D_1(j,n) + D_2(j,n)\\
&\leq&\frac{b(Q_n)}{\vert Q_n\vert} +  (C+D)\frac{\vert \partial^{\diam Q_n} U_j\vert}{\vert U_j\vert} +   C \sum_{P\in {\mathcal P}(Q_n)} \Big \vert \frac{\sharp_P({\mathcal C}\vert_{U_j})}{\vert U_j\vert}-\nu_P\Big \vert,
\end{eqnarray*}
for all $j,n\in \NN$. This proves (\ref{estimate}). Now the main part of the theorem follows readily. One immediately sees that $\Delta(j,n)$ tends to zero if $j$ and $n$ tend (in the right order) to infinity, i.e.
\begin{equation}\label{tozero}
\lim\limits_{n\rightarrow\infty}\lim\limits_{j\rightarrow\infty} \Delta(j,n)=0.
\end{equation}
The triangle inequality shows that
\[\Big \Vert  \frac{F(U_j)}{\vert U_j \vert} - \frac{F(U_m)}{\vert U_m \vert} \Big\Vert = \Big \Vert  \frac{F(U_j)}{\vert U_j \vert} - \hspace{-0.2cm}\sum_{P\in {\mathcal P}(Q_n)}\hspace{-0.2cm} \nu_P \frac{\tilde F(P)}{\vert Q_n\vert}+\hspace{-0.2cm}\sum_{P\in {\mathcal P}(Q_n)}\hspace{-0.2cm} \nu_P \frac{\tilde F(P)}{\vert Q_n\vert}- \frac{F(U_m)}{\vert U_m \vert} \Big\Vert\leq \Delta (j,n) + \Delta (m,n)\]
holds for all $j,m,n\in \mathbb N$. This implies that $(\vert U_j\vert ^{-1}F(U_j))$ is a Cauchy sequence and hence convergent in the Banach space $X$. We use again (\ref{tozero}) to obtain that $\sum_{P\in {\mathcal P}(Q_n)}\nu_P\frac{\tilde F(P)}{\vert Q_n\vert}$ converges to the same limit.
\end{proof}
With the help of the above theorem we are able to give an explicit bound for the distance between the approximants and the limit term.

\begin{Corollary}\label{coro1}
 Let the assumptions of Theorem \ref{ergthm} be fulfilled. Then we have for all $j,n\in \mathbb N$ the estimates
\[\bigg\Vert\overline F - \frac{F(U_j)}{\vert U_j\vert}\bigg\Vert \leq 2 \frac{b(Q_n)}{\vert Q_n\vert} + (C+D)\frac{\vert \partial^{\diam(Q_n)}U_j  \vert}{\vert U_j\vert}+C\sum\limits_{P\in {\mathcal P}(Q_n)} \Big \vert \frac{\sharp_P ({\mathcal C}\vert_{U_j})}{\vert U_j\vert }- \nu_P  \Big \vert\]
and
\[\bigg\Vert \overline F - \sum\limits_{P\in{\mathcal P}(Q_n)} \nu_P \frac{\tilde F(P)}{\vert Q_n\vert} \bigg\Vert\leq \frac{b(Q_n)}{\vert Q_n\vert}.\]
\end{Corollary}
\begin{proof}
By definition of $\overline F$
\[
  \Big\Vert \overline{F}-\frac{F(U_j)}{\vert U_j\vert} \Big\Vert  = \lim\limits_{k\rightarrow\infty}\Big\Vert \frac{F(U_k)}{\vert U_k\vert} -\frac{F(U_j)}{\vert U_j\vert}  \Big\Vert
\]
holds. The triangle inequality implies that for any $n\in \NN$ the above expression is less or equal to
\[
 \lim\limits_{k\rightarrow\infty} \bigg(\Big\Vert \frac{F(U_k)}{\vert U_k\vert} - \sum\limits_{P\in{\mathcal P}(Q_n)}\nu_P \frac{\tilde F(P)}{\vert Q_n\vert}\Big\Vert + \Big\Vert \sum\limits_{P\in{\mathcal P}(Q_n)}\nu_P \frac{\tilde F(P)}{\vert Q_n\vert} - \frac{F(U_j)}{\vert U_j\vert}  \Big\Vert\bigg)= \lim\limits_{k\rightarrow\infty} \big(\Delta (k,n) + \Delta (j,n)\big).
\]
Using the estimate (\ref{estimate}) for $\Delta(j,n)$ we obtain
\begin{eqnarray*}
\Big\Vert \overline{F}-\frac{F(U_j)}{\vert U_j\vert} \Big\Vert&\leq&  \lim\limits_{k\rightarrow\infty} \bigg( \frac{b(Q_n)}{\vert Q_n \vert} +  (C+D)\frac{\vert\partial^{\diam Q_n}U_k \vert}{\vert U_k \vert}+C\sum\limits_{P\in {\mathcal P}(Q_n)}\Big \vert \frac{\sharp_P({\mathcal C}\vert_{U_k})}{\vert U_k\vert}-\nu_P \Big\vert \\
&&\hspace{0.5cm}+ \frac{b(Q_n)}{\vert Q_n \vert} +  (C+D)\frac{\vert\partial^{\diam Q_n}U_j \vert}{\vert U_j \vert}+C\sum\limits_{P\in {\mathcal P}(Q_n)}\Big \vert \frac{\sharp_P({\mathcal C}\vert_{U_j})}{\vert U_j\vert}-\nu_P \Big\vert \bigg) \\
&=& 2\frac{b(Q_n)}{\vert Q_n \vert} +  (C+D)\frac{\vert\partial^{\diam Q_n}U_j \vert}{\vert U_j \vert}+C\sum\limits_{P\in {\mathcal P}(Q_n)}\Big \vert \frac{\sharp_P({\mathcal C}\vert_{U_j})}{\vert U_j\vert}-\nu_P \Big\vert .
\end{eqnarray*}
We use the same ideas to compute the distance to the second estimate:
\begin{eqnarray*}
\bigg\Vert \overline F - \sum\limits_{P\in {\mathcal P}(Q_n)}\nu_P \frac{\tilde F(P)}{\vert Q_n\vert} \bigg\Vert
 &= & \lim\limits_{k\rightarrow \infty} \bigg\Vert \frac{F(U_k)}{\vert U_k\vert} - \sum\limits_{P\in {\mathcal P}(Q_n)}\nu_P \frac{\tilde F(P)}{\vert Q_n\vert} \bigg\Vert  = \lim\limits_{k\rightarrow \infty} \Delta (k,n)\\
& \hspace{-7.5cm}\leq&\hspace{-4cm}\lim\limits_{k\rightarrow \infty} \bigg(\frac{b(Q_n)}{\vert Q_n\vert} + (C+D)\frac{\vert \partial^{\diam(Q_n)}U_k  \vert}{\vert U_k\vert}+ C\sum\limits_{P\in {\mathcal P}(Q_n)} \Big \vert \frac{\sharp_P ({\mathcal C}\vert_{U_k})}{\vert U_k\vert }- \nu_P  \Big \vert\bigg) = \frac{b(Q_n)}{\vert Q_n\vert}.
\end{eqnarray*}
\end{proof}
The assumptions of Theorem \ref{ergthm} and Corollary \ref{coro1} are particularly satisfied if there exists a F\o{}lner sequence  $(Q_n)_{n\in \NN}$ along which the frequencies $\nu_P$ for all patterns $P\in \bigcup_{n\in\NN}{\mathcal P}(Q_n)$ exist and where each $Q_n$ symmetrically tiles $G$. This corresponds to the special case of Assumption  \ref{ass:first}, where it is possible to choose $(U_j)=(Q_j)$.

\begin{Remark}\label{Frechet}
As the proof of Theorem \ref{ergthm} is based on the existence of a Cauchy sequence and its convergence, the completness of the image space $X$ is an essential assumption. However, the existence of a norm on $X$ can be replaced by the existence of a countable family of semi-norms. This gives rise to an induced metric on $X$. Therefore the above statements still hold true in the case were $X$ is a Fr\'echet space.  For similar reasoning see \cite{LenzS-06}.
\end{Remark}


\section{Approximation of the IDS}
We will use Theorem \ref{ergthm} to give an approximation estimate with respect to the supremum norm of the integrated density of states of certain operators. For this purpose we want to describe which class of operators can be treated by our method.
Let ${\mathcal H}$ be a finite dimensional Hilbert space with scalar product $\langle \cdot,\cdot\rangle$ and induced norm $\Vert \cdot \Vert$. We define
\[\ell^2(G,{\mathcal H}):=\left\{u:G\rightarrow {\mathcal H}\ \Big\vert\ \sum\limits_{x\in G}\Vert u(x)\Vert^2 < \infty\right\},\]
which is a Hilbert space as well. For an arbitrary element $x\in G$ let
\begin{equation}\label{proj_pt}
 p_x: \ell^2(G,{\mathcal H}) \rightarrow {\mathcal H},\ u\mapsto p_x(u):= u(x)
\end{equation}
be the natural projection and
\begin{equation}\label{incl_pt}
i_x: {\mathcal H}\rightarrow \ell^2(G,{\mathcal H}),\ h\mapsto i_x(h) \mbox{ with } i_x(h)(y)=h\delta_{xy}
\end{equation}
the natural inclusion. Note that $i_x$ is the adjoint of $p_x$. These maps can be generalised for subsets $Q\subseteq G$. The support of $u\in \ell^2(G,{\mathcal H})$ is the set of those $x\in G$, such that $u(x)\neq 0$. We identify $\ell^2(Q,{\mathcal H})=\{u:Q\rightarrow {\mathcal H} \vert \sum_{x\in Q}\Vert u(x)\Vert^2 < \infty\}$ with the subspace of $\ell^2(G,{\mathcal H})$ consisting of all elements supported in $Q$. The map $p_Q:\ell^2(G,{\mathcal H})\rightarrow \ell^2(Q,{\mathcal H})$ is given by $u\mapsto p_Q(u)$, where $p_Q(u)(x)=u(x)$ for $x\in Q$. Similarly the inclusion $i_Q:\ell^2(Q,{\mathcal H})\rightarrow \ell^2(G,{\mathcal H})$ is given by
\[ i_Q(u)(x):=\left\{\begin{array}{ll} u(x) &\mbox{if $x\in Q$}\\ 0& \mbox{else} \end{array}\right. .\]
In particular we will consider the restricted operator $p_QHi_Q:\ell^2(Q,{\mathcal H})\rightarrow \ell^2(Q,{\mathcal H})$ for which we will use the notation \[H[Q]:=p_QHi_Q.\]

\begin{Definition}
 Let ${\mathcal A}$ be a finite set, ${\mathcal C}:G \rightarrow {\mathcal A}$ a colouring and $H:\ell^2(G,{\mathcal H})\rightarrow \ell^2(G,{\mathcal H})$ a selfadjoint operator. Then we say that
\begin{itemize}
 \item [(a)] $H$ is of \textnormal{finite range} $:\Leftrightarrow \exists M>0$ such that $p_yHi_x=0$ for all $x,y\in G$ with $d_S(x,y)\geq M$,
 \item [(b)] $H$ is \textnormal{${\mathcal C}$-invariant} $:\Leftrightarrow \exists N\in \mathbb N$ such that $p_yHi_x=p_{yt}Hi_{xt}$ for all $x,y,t\in G$ obeying
\[\big(  {\mathcal C}\vert_{B_{N}(x) \cup B_{N}(y)}\big)t = {\mathcal C}\vert_{B_{N}(xt) \cup B_{N}(yt)},\]
 \item [(c)] $R(H):= \max\{M,N\}$ is the \textnormal{overall range} of $H$, if $H$ is of finite range with parameter $M$ and ${\mathcal C}$-invariant with parameter $N$.
\end{itemize}
\end{Definition}
Note that if $H$ satisfies the condition (b) for some $N$, then it does so for any $\tilde N>N$ as well.

\begin{Remark}[Boundedness of $H$]\label{remark_bounded}
 We fix a basis of the Hilbert space ${\mathcal H}$. Since ${\mathcal H}$ is of finite dimension, for each pair $x,y\in G$ the mapping $p_yHi_x: {\mathcal H}\rightarrow {\mathcal H}$ is given by a matrix of dimension $\dim({\mathcal H})\times \dim({\mathcal H})$. The $\mathcal C$-invariance of $H$ implies that the matrix corresponding to $p_yHi_x$ is a function of $\mathcal C$ restricted to $B_N(x)\cup B_N(y)$. If $H$ is of finite range, the matrix is in fact a function of ${\mathcal C}\vert_{B_{2R}(x)}$ only. The reason is that for $x$ and $y$ with distance greater than $M$, $p_yHi_x$ vanishes identically, while for $d_S(x,y)\leq M$ the set $B_N(x)\cup B_N(y)$ is contained in $B_{2R}(x)$. Since $\vert {\mathcal A}\vert <\infty$ and $\vert B_{2R}(x)\vert<\infty$ there are only finitely many functions $P:B_{2R}(x)\rightarrow {\mathcal A}$ and hence only finitely many values which the matrix valued function $p_yHi_x$ can take. From this we conclude that
\begin{equation}\label{bddc}
c:=\sup_{x,y\in G} \sup \{\Vert (p_yHi_x)h\Vert \ \vert\  h\in {\mathcal H}, \Vert h\Vert \leq 1 \}
\end{equation}
is finite. For a given $\phi \in \ell^2(G,{\mathcal H})$ the finite range of $H$ implies $H\phi (x) = \sum_{y\in B_R(x)} (p_yHi_x)\phi(y)$. Hence
\begin{equation*}
 \Vert H\phi \Vert^2 = \langle H\phi, H\phi \rangle =\sum_{x\in G}\langle H\phi (x), H\phi (x) \rangle = \sum_{x\in G} \Big\langle \sum_{y\in B_R(x)}(p_xHi_y) \phi (y) , \sum_{z\in B_R(x)}(p_xHi_z) \phi (z) \Big\rangle
\end{equation*}
holds. The Cauchy-Schwarz inequality implies
\begin{equation*}
 \Vert H\phi \Vert^2 \leq \sum_{x\in G} \sum_{y,z\in B_R(x)}\Vert(p_xHi_y) \phi (y)\Vert \Vert (p_xHi_z) \phi (z) \Vert\leq \sum_{x\in G} \sum_{y,z\in B_R(x)}c^2 \Vert \phi (y) \Vert \Vert \phi(z)\Vert
\end{equation*}
with $c$ as in (\ref{bddc}). Young's inequality $2\Vert\phi (x) \Vert\Vert\phi (y) \Vert\leq \Vert\phi (x) \Vert^2+\Vert\phi (y) \Vert^2$ yields that the last expression is less or equal to
\[\frac{c^2}{2} \sum_{x\in G}\left( \sum_{y,z\in B_R(x)}\Vert \phi (y) \Vert^2+ \sum_{y,z\in B_R(x)}\Vert \phi(z)\Vert^2\right)=c^2\vert B_R\vert \sum_{x\in G} \sum_{y\in B_R(x)}\Vert \phi (y) \Vert^2\]
This shows the boundedness of $H$:
\[\Vert H\phi \Vert\leq c \vert B_R\vert \Vert \phi \Vert.\]
\end{Remark}

\begin{Definition}
 Let ${\mathcal B}(\mathbb R)$ be the Banach space of right-continuous, bounded functions $f:\RR\rightarrow \RR$ equipped with supremum norm. For a selfadjoint operator $A$ on a finite dimensional Hilbert space $V$ we define its cumulative eigenvalue counting function $n(A)\in {\mathcal B}(\mathbb R)$ by setting
\[ n(A)(E):= \vert \{ i\in \mathbb N \ \vert\ \lambda_i\leq E \} \vert\]
for all $E \in \mathbb R$, where $\lambda_i, i=1,\dots,\dim V$ are the eigenvalues of $A$, counted according to their multiplicity.
\end{Definition}
We will study functions $n(H[Q])$ for ${\mathcal C}$-invariant operators $H$ of finite range and $Q\in{\mathcal F}(G)$. Dividing this function by the number of eigenvalues $\dim ({\mathcal H})\vert Q \vert$ of $H[Q]$ gives rise to a distribution function of a probability measure. It encodes the distribution of the spectrum of $H[Q]$. One can hope that if $Q$ tends to the whole group $G$ in some appropriate sense, the normalised eigenvalue counting functions will converge to a limiting distribution.\newline
For this approximation to hold it is certainly a good idea to choose $Q$ from a F\o{}lner sequence $(U_j)_{j\in \NN}$ of subsets in $G$. Indeed, it turns out that the above strategy can be implemented and in fact gives convergence of the distribution functions with respect to supremum norm. This is stated in Theorem \ref{ids_unif} below. We formulate now a convenient assumption.
\begin{assumption}\label{ass:B}
Assume \ref{ass:first} and additionally that $\mathcal H$ is a finite dimensional Hilbert space and that the operator $H:\ell^2(G,{\mathcal H})\rightarrow \ell^2 (G,{\mathcal H})$ is selfadjoint, ${\mathcal C}$-invariant and of finite range. Let $R=R(H)$ denote the overall range of $H$.
\end{assumption}

\begin{Theorem}\label{ids_unif}
 Assume \ref{ass:B}. Then there exists a unique probability measure $\mu_H$ on $\mathbb R$ with distribution function $N_H$, such that the estimate
\[
\bigg\Vert \frac{n(H[U_{j,R}])}{\dim ({\mathcal H})\vert U_{j,R}\vert} -  N_H\bigg\Vert_\infty \leq   8\frac{\vert \partial^RQ_n \vert}{\vert Q_n\vert} + (1+4 \vert B_R\vert )\frac{\vert \partial^{\diam(Q_n)}U_j  \vert}{\vert U_j\vert}+\sum\limits_{P\in {\mathcal P}(Q_n)} \Big \vert \frac{\sharp_P ({\mathcal C}\vert_{U_j})}{\vert U_j\vert }- \nu_P  \Big \vert + \frac{\vert\partial_{int}^{R}U_j\vert}{ \vert U_j\vert }
\]
holds for all $j,n\in\mathbb N$. This implies in particular the convergence
\[ \frac{n(H[U_{j,R}])}{\vert U_{j,R} \vert \dim (\mathcal H)}\rightarrow N_H \]
with respect to the supremum norm for $j\rightarrow \infty$. $N_H$ is called the \textnormal{integrated density of states}.
\end{Theorem}
For the proof we establish a couple of auxiliary results.

\begin{Proposition}\label{prop3}
 Assume \ref{ass:B}. The function $F_R^H: {\mathcal F}(G)\rightarrow {\mathcal B}(\mathbb R),\ Q\mapsto F_R^H(Q):= n(H[Q_R])$ is ${\mathcal C}$-invariant and almost-additive with the boundary term $b(Q):=4\vert \partial^RQ \vert \dim ({\mathcal H})$.
\end{Proposition}
\begin{proof}
Since $R$ is the overall range of $H$ the values of $n(H[Q_R])$ only depend on the colouring of $Q$, namely ${\mathcal C}\vert_{Q}$, and hence $F_R^H$ is ${\mathcal C}$-invariant. To show almost-additivity we use a decoupling argument. Let $Q$ be a disjoint union of $Q_k$ for $k=1,\dots,m$. By definition $R$ is big enough so that
\[H\Big[\bigcup_{k=1}^m Q_{k,R}\Big]=\bigoplus\limits_{k=1}^m H[Q_{k,R}]\]
holds. Therefore we can count the eigenvalues of $ H[Q_{k,R}]$ for $k=1,\dots, m$ separately
\[n\Big(H\Big[\bigcup_{k=1}^m Q_{k,R}\Big]\Big)=\sum\limits_{k=1}^m n(H[Q_{k,R}]).\]
Now we can apply Proposition \ref{prop2} with $V=\ell^2(Q_R,{\mathcal H})$ and $U=\ell^2(\bigcup_{k=1}^m Q_{k,R},{\mathcal H})$. Hence we get
\[ \left\Vert n(H[Q_R])-n\Big(H\Big[\bigcup_{k=1}^m Q_{k,R}\Big]\Big) \right\Vert_\infty \leq 4 \sum_{k=1}^m \vert \partial^RQ_k \vert \dim ({\mathcal H})\]
This proves the almost-additivity with the boundary term $b(Q):= 4\vert \partial^RQ \vert \dim ({\mathcal H})$. Note that $b(Q)\leq D\vert Q\vert$ with $D=4 \vert B_R \vert\dim ({\mathcal H})$, in the notation of Definition \ref{def_boundaryterm}.
\end{proof}
 From a calculation analogous to (\ref{ftildebounded0}) it is clear that $F_R^H$ is bounded. Since the operator $H[Q_R]$ has exactely $\dim({\mathcal H})\vert Q_R\vert$ eigenvalues, the boundedness holds with the constant $C=\dim({\mathcal H})$.
\begin{proof}[Proof of Theorem \ref{ids_unif}]
Since $F_R^H$ is ${\mathcal C}$-invariant and almost-additive, we can apply Corollary \ref{coro1} which gives the existence of a function $\tilde N_H$ with

\begin{multline}\bigg\Vert \frac{F_R^H(U_j)}{\vert U_j\vert} - \tilde N_H\bigg\Vert_\infty \leq 2 \frac{b(Q_n)}{\vert Q_n\vert} + (C+D)\frac{\vert \partial^{\diam(Q_n)}U_j  \vert}{\vert U_j\vert}+C\sum\limits_{P\in {\mathcal P}(Q_n)} \Big \vert \frac{\sharp_P ({\mathcal C}\vert_{U_j})}{\vert U_j\vert }- \nu_P  \Big \vert\\
\leq  \dim ({\mathcal H})\left( 8\frac{\vert \partial^RQ_n \vert}{\vert Q_n\vert} + (1+4 \vert B_R\vert )\frac{\vert \partial^{\diam(Q_n)}U_j  \vert}{\vert U_j\vert}+\sum\limits_{P\in {\mathcal P}(Q_n)} \Big \vert \frac{\sharp_P ({\mathcal C}\vert_{U_j})}{\vert U_j\vert }- \nu_P  \Big \vert\right)
\end{multline}
for all $j,n\in \mathbb N$. What remains to be done is to change the normalisation of $F_R^H(U)=n(H[U_R])$. We know that $\vert U_{j,R}\vert=\vert U_j\vert - \vert \partial^R_{int}(U_j)\vert$ and by expansion one can show that
\[\frac{1}{\vert U_j\vert}=\frac{1}{\vert U_j\vert - \vert \partial^R_{int}(U_j)\vert}- \frac{\vert\partial_{int}^R(U_j)\vert}{\vert U_j\vert (\vert U_j\vert -\vert \partial^R_{int}(U_j)\vert) }\]
holds. This gives us for all $j,n\in \NN$
\begin{eqnarray*}
 \bigg\Vert \frac{F_R^H(U_j)}{\vert U_j\vert} - \tilde N_H\bigg\Vert_\infty &=&\bigg\Vert \frac{F_R^H(U_j)}{\vert U_{j,R}\vert} - \frac{F_R^H(U_j)\vert\partial_{int}^{R}(U_j)\vert}{\vert U_{j,R }\vert  \vert U_j\vert } -\tilde N_H\bigg\Vert_\infty\\
&\geq& \bigg\Vert \frac{F_R^H(U_j)}{\vert U_{j,R}\vert} -\tilde N_H\bigg\Vert_\infty  - \bigg\Vert   \frac{F_R^H(U_j)\vert\partial_{int}^{R}(U_j)\vert}{\vert U_{j,R }\vert  \vert U_j\vert } \bigg\Vert_\infty.
\end{eqnarray*}
By definition of $F_R^H$ we have $\Vert F_R^H(U_j)\Vert_\infty = \dim ({\mathcal H}) \vert U_{j,R}\vert$ which implies
\[ \bigg\Vert   \frac{F_R^H(U_j)\vert\partial_{int}^{R}(U_j)\vert}{\vert U_{j,R }\vert  \vert U_j\vert } \bigg\Vert_\infty = \dim ({\mathcal H})  \frac{\vert\partial_{int}^{R}(U_j)\vert}{ \vert U_j\vert }\]
for all $j\in \NN$. Hence, the estimate
\[\bigg\Vert \frac{F_R^H(U_{j})}{\dim ({\mathcal H})\vert U_{j,R}\vert} -  N_H\bigg\Vert_\infty \leq   8\frac{\vert \partial^RQ_n \vert}{\vert Q_n\vert} + (1+4 \vert B_R\vert )\frac{\vert \partial^{\diam(Q_n)}U_j  \vert}{\vert U_j\vert}+\sum\limits_{P\in {\mathcal P}(Q_n)} \Big \vert \frac{\sharp_P ({\mathcal C}\vert_{U_j})}{\vert U_j\vert }- \nu_P  \Big \vert + \frac{\vert\partial_{int}^{R}U_j\vert}{ \vert U_j\vert }\]
holds for all $j,n\in \NN$ with $N_H:= \tilde N_H / \dim ({\mathcal H})$. By using Lemma \ref{lemma_abc} this shows the claimed convergence.
\end{proof}

Since the eigenvalue counting function $n$ is ${\mathcal C}$-invariant, the function $\tilde n$ on the set the equivalence classes of all patterns  given by
\begin{equation}
\tilde n (\tilde P):=\left \{ \begin{array}{ll} n(H[{Q_R}]) &\mbox{ if } Q\in {\mathcal F}(G) \mbox{ s.t. } \tilde P=\tilde {{\mathcal C}\vert_Q} \\  0& \mbox{ else}
\end{array}   \right.
\end{equation}
is well defined. As before we write $\tilde n(P)$ instead of $\tilde n(\tilde P)$ for a given $P\in \mathcal P$. The following result is a direct consequence of the second estimate in Corollary \ref{coro1} and the boundary term from Proposition \ref{prop3}.

\begin{Corollary}
 Assume \ref{ass:B} and let $N_H$ be defined as in Theorem \ref{ids_unif}. Then the bound
\[\bigg\Vert  N_H - \sum\limits_{P\in{\mathcal P}(Q_n)} \nu_P \frac{\tilde n(P)}{\vert Q_n\vert\dim ({\mathcal H})} \bigg\Vert_\infty\leq 4\frac{\vert\partial^R (Q_n)\vert}{\vert Q_n\vert}\]
holds for all $n\in \mathbb N$.
\end{Corollary}
We give a simple example to show that in general the IDS depends on the choice of the F\o{}lner sequence $(U_j)$.

\begin{Example}\label{ex1}
 Consider the usual graph of $\ZZ$, the set of colours ${\mathcal A}=\{\mbox{black, white}\}$ and the colouring
\[{\mathcal C}:\ZZ \rightarrow {\mathcal A},\hspace{1cm}{\mathcal C}(x)=\left\{\begin{array}{ll} \mbox{white} &\mbox{if $x\geq 0$ or $x=3k$ for $k\in\ZZ$}\\ \mbox{black}& \mbox{else} \end{array}\right.\]
Deleting all edges which are incident with a white vertex gives rise to a new graph and hence a new adjacency operator $A$. This operator is selfadjoint, of finite range and $\mathcal C$-invariant. We choose two F\o{}lner sequences $(U_j)$ and $(V_j)$ as follows
\begin{equation}\label{U_j,V_j}
U_j=\{1,\dots,3j\}\hspace{1cm} \mbox{and}\hspace{1cm} V_j=\{-3j,\dots,-1\}.
\end{equation}

Since for all $j\in\NN$ all entries of the matrix $A[U_j]$ are equal to zero, the IDS $N_U$ with respect to the sequence $(U_j)$ is
\[N_U(\lambda)=\left\{\begin{array}{ll} 0 &\mbox{if $\lambda <0$}\\ 1& \mbox{else.} \end{array}\right. \]
Computing the IDS along the sequence $(V_j)$ gives a completely different picture: The eigenvalues of the matrix $A[V_j]$ are $-1,0$ and $1$, each of them with multiplicity $j$. Therefore the IDS $N_V$ with respect to the sequence $(V_j)$ is the function
\[ N_V(\lambda)=\left\{\begin{array}{ll} 0 &\mbox{if $\lambda <-1$}\\ 1/3 &\mbox{if $-1\leq \lambda <0$}\\2/3 &\mbox{if $0\leq\lambda <1$}\\1& \mbox{else.} \end{array}\right. \]
\end{Example}

\begin{assumption}\label{ass:C}
  Assume \ref{ass:B} and additionally that the frequencies $\nu_P$ are strictly positive for all patterns $P\in {\mathcal P}$ which occur in ${\mathcal C}$, i.e. for which there exists $g\in G$ with ${\mathcal C}\vert_{D(P)g}= Pg $
\end{assumption}

\begin{Theorem}
 If we assume \ref{ass:C}, then the spectrum of $H$ equals the topological support of $\mu_H$.
\end{Theorem}
\begin{proof}
Since the operator $H$ is assumed to be of overall range $R$ we have
\begin{equation}\label{bew_thm2_1}
\Vert  (H-\lambda)u \Vert=\Vert  (H[Q]-\lambda)p_Q u \Vert,
\end{equation}
for all $u$ with $\supp (u)\subseteq Q_R$.
Let $\lambda$ be an element of the spectrum $\sigma(H)$, then $H-\lambda$ is not invertible. Thus, for each $\epsilon>0$ we can find a subset $Q\in {\mathcal F}(G)$ and a normalised vector $u$ with support in $Q_R$ such that $\Vert(H-\lambda)u \Vert <\epsilon$ holds. From this we know that $(H-\lambda)u$ is supported in $Q$ and $\Vert(H[Q]-\lambda)p_Q u \Vert <\epsilon$ by (\ref{bew_thm2_1}). For each $j\in \mathbb N$ we denote the number of disjoint occurrences of translates of ${\mathcal C}\vert_{Q}$ in the set $U_{j,R}$ by $k(j)$. This ensures the existence of $k(j)$ pairwise orthogonal normalised vectors $u_i$, $i=1,\dots,k(j)$, where also $(H-\lambda)u_i$ are pairwise orthogonal and of norm strictly less than $\epsilon$. Applying Proposition \ref{prop4} we get that there must be at least $k(j)$ eigenvalues in the interval $(\lambda-\epsilon, \lambda+\epsilon)$ i.e.
\[n(H[U_{j,R}])(\lambda+\epsilon)-n(H[U_{j,R}])(\lambda-\epsilon)\geq k(j).\]
For a pattern $P\in \mathcal P$ Lemma \ref{lemma_abc} yields that the frequency $\nu_P$ along $(U_j)$ is the same as the frequency along $(U_{j,R})$. As these frequencies are assumed to be strictly positive for all patterns which occur in $\mathcal C$, the number of disjoint occurrences of ${\mathcal C}\vert_{Q}$ in $U_{j,R}$ grows linearly in the volume of $U_{j,R}$ for large $j$. Thus we can find a $c>0$ such that  $k(j)\geq c\vert U_{j,R} \vert$ holds for large $j$. Using the uniform convergence of $n(H[{U_{j,R}}])$ we see
\[\mu_H([\lambda-\epsilon,\lambda+\epsilon])=\lim_{j\rightarrow\infty}\frac{n(H[U_{j,R}])(\lambda+\epsilon)-n(H[U_{j,R}])(\lambda-\epsilon)}{\vert U_{j,R}\vert \dim {(\mathcal H)}}\geq \lim_{j\rightarrow\infty}\frac{k(j)}{\vert U_{j,R}\vert \dim {(\mathcal H)}}\geq \frac{c}{\dim {(\mathcal H)}}.\]
As $c$ is strictly positive and $\epsilon>0$ was arbitrary, we conclude that $\lambda$ is in the support of $\mu_H$.

\medskip
Now we start with $\lambda$ in the support of $\mu_H$. Thus for each $\epsilon>0$ we have a $c>0$ such that $\mu_H([\lambda-\epsilon,\lambda+\epsilon])\geq c$. By uniform convergence this gives that
\[ n(H[{U_{j,R}}])(\lambda+\epsilon)-n(H[{U_{j,R}}])(\lambda-\epsilon)\geq\frac{c}{2}\dim ({\mathcal H})\vert U_{j,R}\vert \]
holds for large $j$. Now we use Proposition \ref{prop2} to observe
\[\Vert n(H[U_{j,2R}])-n(H[U_{j,R}])\Vert_\infty \leq 4  \dim ({\mathcal H})\vert \partial_{int}^R U_{j,R}\vert\]
which leads together with triangle inequality to
\[ n(H[{U_{j,2R}}])(\lambda+\epsilon)-n(H[{U_{j,2R}}])(\lambda-\epsilon)\geq\frac{c}{2}\dim ({\mathcal H})\vert U_{j,R}\vert-8\vert \partial_{int}^{R}U_{j,R}\vert \dim({\mathcal H})\]
for large $j$. As the right hand side is positive for large $j$, there exists an eigenvalue $\overline\lambda\in [\lambda-\epsilon,\lambda+\epsilon]$ and a normalised eigenvector $\overline u\in \ell^2(U_{j,R},\mathcal H)$ such that $(H[{U_{j,2R}}]-\overline \lambda)\overline u=0$ holds. From this we have
\begin{equation*}
 \Vert (H[{U_{j,R}}]- \lambda)p_{U_{j,R}}u\Vert=\Vert (H[U_{j,R}]-\overline \lambda)p_{U_{j,R}}u+(\overline\lambda-\lambda)p_{U_{j,R}}u\Vert \leq \vert \overline\lambda-\lambda\vert \leq\epsilon
\end{equation*}
with a normalised vector $u=i_{U_{j,2R}}\overline{u}\in\ell^2(G,\mathcal H)$ which is supported in $U_{j,2R}$. By (\ref{bew_thm2_1}) we get $\Vert (H-\lambda)u\Vert\leq \epsilon$ and $\sigma(H)\cap[\lambda-\epsilon,\lambda+\epsilon]\neq \emptyset$. Since $\epsilon>0$ is arbitrary, we infer that $\lambda$ belongs to $\sigma(H)$.
\end{proof}
The following example shows that the positivity of the frequencies is a necessary assumption.

\begin{Example}
 Consider the same situation as in Example \ref{ex1} but now choose the colouring
\[{\mathcal C}:\ZZ \rightarrow {\mathcal A},\hspace{1cm}{\mathcal C}(x)=\left\{\begin{array}{ll} \mbox{white} &\mbox{if $x\geq 0$ or $x\leq-100$ or $x=3k$ for $k\in\ZZ$}\\ \mbox{black}& \mbox{else.} \end{array}\right.\]
Again we treat the case where edges only exist between black vertices with distance one. The restricted adjacency operator $A[V_j]$, with $V_j$ as in (\ref{U_j,V_j}), has for $1\leq j\leq 33$ the eigenvalues $-1,0$ and $1$ each of them with multiplicity $j$. From this we get in particular that $-1$ and $1$ are elements of the spectrum of $A$.\newline However the frequencies of the patterns that give rise to these eigenvalues is zero. For all $j\geq 34$ the multiplicities of the eigenvalues $-1$ and $1$ remain $33$ and the multiplicity of the eigenvalue $0$ equals $33+3j$. Therefore for increasing $j$ the steps of the cumulative eigenvalue counting become relatively small. This implies that the IDS is the function
\[N(\lambda)=\left\{\begin{array}{ll} 0 &\mbox{if $\lambda <0$}\\ 1& \mbox{else.} \end{array}\right. \]
Thus the topological support of the induced measure $\mu_A$ equals $\{0\}$, though $-1$ and $1$ are in the spectrum of $A$.
\end{Example}

The next corollary characterises the set of points at which the IDS is discontinuous. It has been obtained previously in \cite{LenzV-09} by different methods. For earlier results characterising the set of jumps see e.g.
\cite{KlassertLS-03}, \cite{Veselic-05b}. Related results pointing in this direction have been observed already in
\cite{Kuchment-91,Kuchment-05}.

\begin{Corollary}\label{cor_disc}
 Assume \ref{ass:C} and let $\lambda\in \mathbb R$. Then the following assertions are equivalent:
\begin{itemize}
 \item [(i)] $\lambda$ is a point of discontinuity of $N_H$,
 \item [(ii)] there exists a compactly supported eigenfunction of $H$ corresponding to $\lambda$.
\end{itemize}
\end{Corollary}
\begin{proof}
 Let $u$ be an eigenfunction corresponding to $\lambda$ with $r>0$ such that $\supp (u)\subseteq B_r$ holds. Let $Q$ be a finite subset of $G$. Set $P:={\mathcal C}\vert_{B_r}$, then each copy of $P$ in ${\mathcal C}\vert_{Q}$ adds a dimension to the eigenspace of $p_QHi_Q$ belonging to $\lambda$. We denote the number of disjoint copies of $P$ in $Q$ by $\dot\sharp_P({\mathcal C}\vert_{Q})$. A simple combinatorial argument shows $\vert B_{3r} \vert\dot\sharp_P({\mathcal C}\vert_{Q})\geq \sharp_P({\mathcal C}\vert_{Q})$. With this we get
\[\frac{n(H[Q])(\lambda-\epsilon)}{\vert Q\vert}\leq \frac{n(H[Q])(\lambda+\epsilon)-\dot\sharp_P({\mathcal C}\vert_{Q})}{\vert Q\vert}\leq \frac{n(H[Q])(\lambda+\epsilon)}{\vert Q\vert}-\frac{\sharp_P({\mathcal C}\vert_{Q})}{\vert B_{3r}\vert \vert Q\vert}.\]
Now we substitute $Q$ by the elements of the F\o{}lner sequence $(U_{j,R})_{j\in \mathbb N}$
\[\frac{n(H[{U_{j,R}}])(\lambda+\epsilon)}{\vert U_{j,R}\vert} -\frac{n(H[{U_{j,R}}])(\lambda-\epsilon)}{\vert U_{j,R}\vert}\geq \frac{\sharp_P({\mathcal C}\vert_{U_{j,R}})}{\vert B_{3r}\vert \vert U_{j,R} \vert}.\]
If $j\rightarrow \infty$ we get
\begin{equation}\label{jumpest}
 N_H(\lambda+\epsilon)-N_H(\lambda-\epsilon)\geq \frac{\nu_P}{\vert B_{3r}\vert \dim ({\mathcal H})}>0,
\end{equation}
where we used Lemma \ref{lemma_abc}. As $\epsilon>0$ is arbitrary, we infer that $\lambda$ is a point of discontinuity of $N_H$.

\medskip

 Let $\lambda$ be a point of discontinuity of $N_H$ and $(U_j)_{j\in \mathbb N}$ F\o{}lner sequence given in (\ref{ass:B}). Theorem \ref{ids_unif} shows that the distribution function $n(H[{U_{j,R}}])/(\vert U_{j,R}\vert\dim (\mathcal H))$ converges to $N_H$ with respect to the supremum norm. Since $\lambda$ is a point of discontinuity the jump at $\lambda$ will not get small, i.e.
\[\dim (\ker(H[{U_{j,R}}]-\lambda))=\lim\limits_{\epsilon\rightarrow 0}(n(H[{U_{j,R}}])(\lambda+\epsilon)-n(H[{U_{j,R}}])(\lambda-\epsilon))\geq c\vert U_{j,R}\vert\]
for a $c>0$ and all $j\in \mathbb N$. We also know
\[\dim (\ell^2(\partial^{2R}_{int}U_{j,R}))=\vert \partial^{2R}_{int}U_{j,R} \vert= \frac{\vert\partial^{2R}_{int}U_{j,R}\vert}{\vert U_{j,R}\vert}\vert U_{j,R}\vert\]
and since $U_{j,R}$ is a F\o{}lner sequence $\frac{\vert\partial^{2R}_{int}U_{j,R}\vert}{\vert U_{j,R}\vert}$ tends to zero for large $j$. Thus we get that
\[\dim (\ker(H[{U_{j,R}}]-\lambda))>\dim (\ell^2(\partial^{2R}_{int}U_{j,R}))\]
holds for large $j$. By using Proposition \ref{prop_help} we find an eigenvector $u$ of $H$ with $\supp u\subseteq U_{j,3R}$ for some $j\in \mathbb N$.
\end{proof}

\begin{Remark}
Part of Corollary \ref{cor_disc} can be sharpened to obtain a quantitative estimate. First we will discuss the size of a jump of the IDS at a point of discontinuity as it is estimated in line (\ref{jumpest}). If $v$ is an eigenfunction with support in $B_r$ then any pattern which is equivalent to $\mathcal C\vert_{B_r}$ will support an eigenfunction which is a translate of $v$. The frequency describes how often the respective pattern occurs in the colouring. In this sense it describes also the density of translates of the eigenfunction $v$. Thus it is natural that the size of the jump is proportional to $\nu_P$.
Since the IDS measures the number of eigenstates per unit volume,
the size of the ball $B_r$ and the dimension of the auxiliary Hilbert space $\mathcal H$ enter in the
denominator of the size of the jump. In fact there the term $\vert B_{3r}\vert$ rather then $\vert B_{r}\vert$ occurs since
in the proof of Corollary \ref{cor_disc} we are interested in the number of \emph{disjoint} occurrences
of a pattern in a certain volume, whereas the frequency $\nu_P$ is defined
in terms of overall occurrences of the sought pattern.

An analogous reasoning shows that if we have $m$ linearly independent eigenfunctions in $\ell^2(B_r)$ to the energy $\lambda$
then the size of the jump will obey
\[
N_H(\lambda+\epsilon)-N_H(\lambda-\epsilon)\geq
\frac{m \, \nu_P}{\vert B_{3r}\vert \dim(\mathcal H)},
\]
for any $\epsilon >0$.
\end{Remark}


\section{Applications} \label{section_app}
\subsection{Abelian groups}
In this section the main results are applied to the case where the group $G$ equals $\ZZ^d$, as an example for a finitely generated abelian group. 
Let $S$ be the usual set of generators given by $S=\{\pm s_1,\dots, \pm s_d\}$ with $s_i(j)=\delta_{ij}$. It is easy to check that the sequence $(Q_n)$ of cubes $Q_n=\{ 0,\dots, n-1 \}^d$ is a F\o{}lner sequence. Moreover for each $n\in \NN$ the set $Q_n$ symmetrically tiles $\ZZ^d$ with grid $(n\ZZ)^d$.
\begin{assumption}\label{ass:D}
 Let $(Q_n)$ and $S$ be as above. $\mathcal A$ is a finite set of colours, ${\mathcal C}: G\rightarrow {\mathcal A}$ a map called colouring and $(U_j)$ a F\o{}lner sequence along which the frequencies of all patterns $P\in \bigcup_{n\in \NN} {\mathcal P}(Q_n)$ exist.
\end{assumption}
One obtains the following corollary as a special case of Theorem \ref{ergthm} by using the equalities
\begin{equation}\label{roughest1}
 \vert Q_n\vert = n^d\hspace{0.7cm}\mbox{and}\hspace{0.7cm} \diam(Q_n) = dn.
\end{equation}

The following Corollary recovers the main result of \cite{LenzMV-08}.

\begin{Corollary}
  Assume \ref{ass:D}. For a given ${\mathcal C}$-invariant and almost-additive function $F: {\mathcal F}(\ZZ^d)\rightarrow X$ the following limits
\[\overline F := \lim\limits_{j\rightarrow \infty}\frac{F(U_j)}{\vert U_j \vert}= \lim\limits_{n \rightarrow\infty} \sum_{P\in {\mathcal P}(Q_n)} \nu_P \frac{\tilde F(P)}{n^d}  \]
exist and are equal. Furthermore, for $j,n\in \mathbb N$ the difference
\[\Delta(j,n):=\Big\Vert \frac{F(U_j)}{\vert U_j \vert} -   \sum_{P\in {\mathcal P}(Q_n)} \nu_P \frac{\tilde F(P)}{n^d}\Big\Vert\]
satisfies the estimate
\begin{equation}
\Delta(j,n)\leq \frac{b(Q_n)}{n^d} + (C+D)\frac{\vert \partial^{nd}U_j  \vert}{\vert U_j\vert}+C\sum\limits_{P\in {\mathcal P}(Q_n)} \Big \vert \frac{\sharp_P ({\mathcal C}\vert_{U_j})}{\vert U_j\vert }- \nu_P  \Big \vert.
\end{equation}
\end{Corollary}

\begin{assumption}\label{ass:E}
Assume \ref{ass:D} and additionally that $\mathcal H$ is a Hilbert space of dimension $k<\infty$ and $H:\ell^2(\ZZ^d, {\mathcal H})\rightarrow \ell^2(\ZZ^d, {\mathcal H})$ a selfadjoint, $\mathcal C$-invariant operator of finite range with overall range $R$.
\end{assumption}
By Theorem \ref{ids_unif} the Assumption \ref{ass:E} implies that there exists a unique distribution function $N_H$ such that the estimate
\[\bigg\Vert \frac{n(H[U_{j,R}])}{k\vert U_{j,R}\vert} -  N_H\bigg\Vert_\infty \leq   8\frac{\vert \partial^RQ_n \vert}{\vert Q_n\vert} + (1+4 \vert B_R\vert )\frac{\vert \partial^{\diam(Q_n)}U_j  \vert}{\vert U_j\vert}+\sum\limits_{P\in {\mathcal P}(Q_n)} \Big \vert \frac{\sharp_P ({\mathcal C}\vert_{U_j})}{\vert U_j\vert }- \nu_P  \Big \vert + \frac{\vert\partial_{int}^{R}U_j\vert}{ \vert U_j\vert }\]
holds for all $j,n\in \NN$. Using the equalities (\ref{roughest1}) and the inequalities
\begin{equation}\label{roughtest}
\vert B_R\vert \leq (2R)^d\hspace{0.7cm}\mbox{and}\hspace{0.7cm} \vert \partial^R (Q_n) \vert \leq (n+4R)^d-n^d
\end{equation}
leads to a slightly weaker corollary.
\begin{Corollary}\label{cor_app1}
 Assume \ref{ass:E}. Then there exists a unique distribution function $N_H$, such that $n(H[U_{j,R}])/(k\vert U_{j,R} \vert)$ converges to $N_H$ with respect to the supremum norm as $j\rightarrow \infty$. In fact, the estimate
\begin{equation*}\label{weaker1}
 \Big\Vert \frac{n(H[U_{j,R}])}{k\vert U_{j,R}\vert}-N_H \Big\Vert_\infty\leq 8 \left(\hspace{-0.1cm}\left(1+\frac{4R}{n}\right)^d\hspace{-0.2cm}-1\hspace{-0.1cm}\right)+(1+4(2R)^d )\frac{\vert\partial^{dn}U_j\vert}{\vert U_j\vert} +\hspace{-0.3cm}\sum\limits_{P\in {\mathcal P}(Q_n)}\Big\vert\frac{\sharp_P({\mathcal C}\vert_{U_j})}{\vert U_j\vert}-\nu_P\Big\vert + \frac{\vert\partial_{int}^{R}U_j\vert}{ \vert U_j\vert }
\end{equation*}
holds for all $j,n\in\mathbb N$.
\end{Corollary}
In the situation where the frequencies $\nu_P$ of all patterns $P\in \bigcup_{n\in \NN} {\mathcal P}(Q_n)$ exist along the sequence of cubes $(Q_n)$ we set $U_j:=Q_j$ for all $j\in \NN$. Again by using (\ref{roughest1}) and (\ref{roughtest}), the estimate in Corollary \ref{cor_app1} can be replaced by
\begin{eqnarray*}
\Big\Vert \frac{n(H[Q_{j,R}])}{k\vert Q_{j,R}\vert}-N_H \Big\Vert_\infty&\leq& 8 \left(\left(1+\frac{4R}{n}\right)^d-1\right)+(1+4(2R)^d )\left(\left(1+\frac{4dn}{j}\right)^d-1\right)\\
&&+\sum\limits_{P\in {\mathcal P}(Q_n)}\Big\vert\frac{\sharp_P({\mathcal C}\vert_{Q_j})}{\vert Q_j\vert}-\nu_P\Big\vert + \left(\left(1+\frac{4R}{j}\right)^d-1\right).
\end{eqnarray*}
If furthermore $\mathcal A$ consists of only one element, all information given by a pattern $P\in \mathcal P$ is its domain $D(P)$. Therefore in this situation the frequencies $\nu_P$ exist for all patterns  $P\in\mathcal P$ along any F\o{}lner sequence $(U_j)$. In fact
\begin{equation}\label{nu=1}
1\geq \frac{\sharp_P({\mathcal C}\vert_{Q_j})}{\vert Q_j\vert}\geq\frac{\vert Q_{j,\diam D(P)}\vert}{\vert Q_j\vert}\rightarrow 1\hspace{0.7cm}\mbox{for }j\rightarrow\infty
\end{equation}
holds and hence $\nu_P=1$ for all $P\in\mathcal P$. Note that $\mathcal P (Q_n)$ contains just one element. In this situation we get
\begin{equation}\label{onecolest}
\sum\limits_{P\in {\mathcal P}(Q_n)}\Big\vert\frac{\sharp_P({\mathcal C}\vert_{Q_j})}{\vert Q_j\vert}-\nu_P\Big\vert
\leq 1-\frac{\vert Q_{j,\diam Q_n}\vert}{\vert Q_j\vert}
\leq \frac{\vert\partial^{\diam Q_n}Q_j\vert}{\vert Q_j\vert}
\leq\left(\left(1+\frac{ 4dn}{j}\right)^d-1\right)
\end{equation}
and hence that the estimate
\[
\Big\Vert \frac{n(H[U_{j,R}])}{k\vert U_{j,R}\vert}-N_H \Big\Vert_\infty \leq c\left(\left(1+\frac{c}{n}\right)^d+ \left(1+\frac{cn}{j}\right)^d-2\right)
\]
holds for all $j,n\in\mathbb N$, where $c=6(2R)^d$.

\subsection{Heisenberg group}
The discrete Heisenberg group $H_3$ is a prominent example for a non-abelian, finitely generated group. A finite system of generators $S$ gives rise to the Cayley graph and the adjacency operator. We are interested in the spectral distribution of this operator. Applying Theorem \ref{ids_unif} leads to a uniform approximation of the IDS. The elements of the discrete Heisenberg group are given by the set
\[ H_3:= \left\{ \begin{pmatrix}
          1 & 0 & 0\\
          a & 1 & 0\\
          c & b & 1
         \end{pmatrix}  \Bigg {\vert} a,b,c \in \mathbb Z \right\} .\]
 We denote an element of $H_3$ by
\[(a,b,c):=  \begin{pmatrix}
          1 & 0 & 0\\
          a & 1 & 0\\
          c & b & 1
         \end{pmatrix}.\]
The group multiplication is induced by the usual matrix multiplication, therefore the product and the inverse for two elements $(a,b,c) , (a',b',c')\in H_3$ are given by
\[(a,b,c)(a',b',c')=(a+a',b+b',c+c'+ba')\mbox{ and } (a,b,c)^{-1}=(-a,-b,ab-c).\]
It is easy to check, that the set $S:=\{s_1^{\pm 1},s_2^{\pm 1}\}$ with $s_1=(1,0,0),s_2=(0,1,0)$ is symmetric and generates $H_3$.
Let the sequence of subgroups $(G_n)$ be given by $G_n:=\{(a,b,c) \vert a,b\in n\mathbb Z, c\in n^2\mathbb Z\}$. One can show that for each $n\in \mathbb N$ the set $Q_n=\{(a,b,c)\vert a,b\in \mathbb Z(0,n-1),c\in \mathbb Z(0,n^2-1)\}$ is a fundamental domain for $G_n$ in $H_3$, where we use for $u\leq v\in\mathbb Z$ the notation $\mathbb Z(u,v):=\{u,u+1,\dots, v\}$. Now we prove that $(Q_n)_{n\in \mathbb N}$ is a F\o{}lner sequence. Since the equality
\[\vert Q_n S\setminus Q_n\vert=\sum\limits_{s\in S}\vert Q_n s\setminus Q_n\vert\]
holds, we study the size of the four disjoint parts of the boundary $\vert Q_n s\setminus Q_n\vert$, $s\in S$ separately. For the first part we get
\begin{eqnarray*}\left\vert Q_n s_1\setminus Q_n\right\vert&=&\left\vert \left\{ (a+1,b,c+b)\big\vert a=n-1, b\in \mathbb Z_0^{n-1},c\in \mathbb Z_0^{n^2-1} \right\}\right\vert\\ &&+\left\vert \left\{ (a+1,b,c+b)\big\vert a\in \mathbb Z_0^{n-2},\ b\in \mathbb Z_0^{n-1},\ c\in \mathbb Z_0^{n^2-1} , b+c\geq n^2\right\}\right\vert\\
&=&n^3+ (n-1)\sum\limits_{b=0}^{n-1}\sum\limits_{c=0}^{n^2-1}\textbf 1_{\{b+c\geq n^2\}}\\
&=&\frac{3}{2}n^3-n^2+\frac{1}{2}n,
\end{eqnarray*}
where we used $\sum_{i=1}^n i= \frac{1}{2}n(n-1)$.
Similarly one can show
\[\vert Q_n s_1^{-1}\setminus Q_n\vert=\frac{3}{2}n^3-n^2+\frac{1}{2}n\hspace{0.6cm} \mbox{and}\hspace{0.6cm}\vert Q_n s_2\setminus Q_n\vert=\vert Q_n s_2^{-1}\setminus Q_n\vert=n^3\]
for the other generators. Therefore we have
\[\vert Q_n S\setminus Q_n\vert=5n^3-2n^2+n\]
for the boundary of a set $Q_n$. As the volume of the fundamental domain $Q_n$ is equal to $n^4$ we get
\[\lim\limits_{n\rightarrow \infty} \frac{\vert Q_n S\setminus Q_n\vert}{\vert Q_n\vert}=0.\]
Thus, the sequence $(Q_n)$ is a F\o{}lner sequence and $H_3$ is amenable.
We consider the trivial colouring, i.e. $\vert {\mathcal A}\vert=1$ on $H_3$. In this case $\nu_P=1$ for all patterns $P\in {\mathcal P}$, cf. (\ref{nu=1}). The adjacency operator $A:\ell^2(H_3)\rightarrow \ell^2(H_3)$ is defined pointwise for $x,y\in H_3$ and $f\in \ell^2(H_3)$ by
\[
A f(x) =\sum\limits_{y\in H_3}A(x,y)f(y),
\hspace{0.5cm}\mbox{ where }\hspace{0.5cm}
A(x,y):=\left\{ \begin{array}{ll}1 & \mbox{ if } d_S(x,y)=1\\0 & \mbox{ else. }\end{array} \right.
\]
This operator is obviously selfadjoint, of finite range and ${\mathcal C}$-invariant with overall range $R(A)=1$. Hence, the Assumption \ref{ass:B} is fulfilled and we can apply Theorem \ref{ids_unif}, which gives us the uniform convergence of the cumulative eigenvalue counting function and the estimate
\begin{eqnarray*}
\Big\Vert \frac{n(A[Q_{j,1}])}{\vert Q_{j,1}\vert}-N_A \Big\Vert_\infty&\leq& 8 \frac{\vert \partial^1  (Q_n)\vert}{\vert Q_n \vert}+(1+4\vert B_1 \vert )\frac{\vert\partial^{\diam (Q_n)}Q_j\vert}{\vert Q_j\vert}\\
&&+\sum\limits_{P\in {\mathcal P}(Q_n)}\Big\vert\frac{\sharp_P({\mathcal C}\vert_{Q_j})}{\vert Q_j\vert}-\nu_P\Big\vert + \frac{\vert\partial_{int}^{1}(Q_j)\vert}{ \vert Q_j\vert } \\
& \leq & 8 \frac{\vert \partial^1  (Q_n)\vert}{\vert Q_n \vert}+22\frac{\vert\partial^{\diam (Q_n)}Q_j\vert}{\vert Q_j\vert} +\sum\limits_{P\in {\mathcal P}(Q_n)}\Big\vert\frac{\sharp_P({\mathcal C}\vert_{Q_j})}{\vert Q_j\vert}-\nu_P\Big\vert
\end{eqnarray*}
Here we used that the ball of radius one contains exactly five elements. Since there exists only one pattern with domain $Q_n$ the last sum is not larger than $ \vert Q_j\vert^{-1} \vert \partial^{\diam(Q_n)}Q_j\vert  $, cf. (\ref{onecolest}). Thus,
\[ \Big\Vert \frac{n(A[Q_{j,1}])}{\vert Q_{j,1}\vert}-N_A \Big\Vert_\infty\leq  8 \frac{\vert \partial^1  (Q_n)\vert}{\vert Q_n \vert}+23\frac{\vert\partial^{\diam (Q_n)}Q_j\vert}{\vert Q_j\vert} \]
holds for all $j,n\in\NN$.

\begin{Remark}[Diameters of $Q_n$]
Up to now we have discussed two features of the sequence $(Q_n)$: Each of its elements symmetrically tiles $G$
and  the sequence has the F\o lner property. Now, the latter  is  a statement about the quotient of the area of the boundary of $Q_n$
and the volume of $Q_n$. For certain questions  it is desireable to relate this quantities to the diameter of $Q_n$.
From the above calculations and reference \cite{Blachere-03a} we conclude that
\[
n\le \diam (Q_n) \le 6 n , \quad
|Q_n| = n^4 , \quad
|Q_nS \setminus Q_n| =5  n^3-2n^2 +n\le 6 n^3.
\]
Using the isoperimetric inequality \cite{Foelner-55,Zuk-00} (valid for any finite subset $Q$ of an arbitrary Cayley graph of the group
$G$ with generating set $S$) and Lemma \ref{lemma_abc} we obtain
\[
 |Q| \le \sum_{\gamma \in \partial_{ext}^1 Q} d_S(e,\gamma)
\le |S|\ |Q S \setminus Q| \, (1 + \diam (Q))
\]
Thus it follows that for all $Q_n$ we have a tight  upper and lower bound for the above mentioned quotient (as far as the
power of $n$ is concerned):
\[
 \frac{1}{28\, n}  \le \frac{1}{|S| (1 + \diam (Q_n))}\le \frac{|Q_nS \setminus Q_n| }{|Q_n| } \le \frac{6}{n}
\]
In some sense our choice of the sequence $(Q_n)$ is optimal, at least if one looks only at the order of powers of $n$.
If we minimize the quotient over all subsets $Q$ of $G$ with diameter $n$ we obtain
\[
 \inf_{\diam Q =n}  \frac{|QS \setminus Q| }{|Q| }  \ge \frac{1}{4+4n} .
\]
Note also that the tiles $Q_n$ have the same relation between diameter and volume as balls of radius $n$:
\[
 |Q_n| \sim  (\diam Q_n)^4.
\]
This implies that for lower bounds on Laplace eigenvalues in terms of Cheeger inequalities  the tiles $Q_n$
are equaly good as balls.

\end{Remark}

\subsection{Periodic operators}
An important class of operators to which our theory can be applied are those which are invariant with respect to a periodic colouring. Here are the details.\newline
Let $G$ be a finitely generated group, $\Gamma$ a locally finite graph with a countable set of vertices and $T$ a representation of $G$ by graph isomorphisms $T_g: \Gamma\rightarrow \Gamma, g\in G$.  We furthermore assume that the action $T$ of $G$ on $\Gamma$ is free and cocompact, where free means that for any distinct $g,h\in G$ and all $\gamma\in \Gamma$ we have $T_g \gamma\neq T_h \gamma$. By the cocompactness assumption we have that the quotient space $\Gamma/T$ is compact and in this case even finite.\newline
A fundamental domain ${\mathcal D}\subseteq \Gamma$ contains by definition exactly one element of each equivalence class of $\Gamma$ with respect to $T$. Therefore $\mathcal D$ is also finite and this implies that the vertex degree of $\Gamma$ is bounded.
Consider for example the Heisenberg group $H_3$ acting on the Cayley graph of $H_3$ with respect to a fixed generating system, where the action of $T_g$ is given by
\[T_g (\gamma) = \gamma g^n \]
for a fixed $n\in \mathbb N$, all $\gamma\in H_3$ and all $g \in H_3$.\newline
We will use the relation between the group $G$ and the graph $\Gamma$ to convert an operator defined on $\ell^2(\Gamma)$ to an operator on $\ell^2(G,\mathcal H)$ for an appropriately chosen finite dimensional Hilbert space $\mathcal H$. Let $A:\ell^2(\Gamma)\rightarrow \ell^2(\Gamma)$ be linear, selfadjoint and invariant under the action of $T$, i.e.
\begin{equation}\label{transinvA}
A(x,y)=A(T_g x ,T_g y) \hspace{2cm}\forall x,y\in \Gamma, g \in G .
\end{equation}
We furthermore assume that $A$ is of finite range, which means that whenever the graph distance between $x$ and $y$ is larger than a constant $\rho$, we have $A(x,y)=0$.\newline
As the fundamental domain $\mathcal D$ is finite, $\mathcal H:=\ell^2({\mathcal D})$ is of finite dimension. Next we define a unitary operator $U:\ell^2(G,{\mathcal H}) \rightarrow \ell^2(\Gamma)$. For any elements $\psi\in \ell^2(G,{\mathcal H})$ and $ g\in G$ we can decompose $\psi( g)$ according to
\[\psi( g)=\sum\limits_{i \in {\mathcal D}}\psi_i(g)\delta_i ,\]
where $\delta_i,\ i\in {\mathcal D}$ are the basis elements of $\ell^2({\mathcal D})$. It is clear that the coefficients $\psi_i(g)$ are uniquely determined. For a given $\gamma\in \Gamma$ we choose the elements $i\in {\mathcal D}$ and $g \in G$ such that $\gamma=T_g i$ holds, to define $U\psi(\gamma):=\psi_i(g)$. It is not hard to check that the operator $U^*:\ell^2(\Gamma)\rightarrow \ell^2(G,{\mathcal H})$ given by $U^*\phi (g)=\sum_{i \in {\mathcal D}}\phi(T_g i)\delta_i$ for $\phi \in \ell^2(\Gamma)$ and $g \in G$ is the inverse and the adjoint of $U$.\newline
Finally we define the operator $H:=U^*AU:\ell^2(G,{\mathcal H})\rightarrow \ell^2(G,{\mathcal H})$ and need to show the $\mathcal C$-invariance and finite range property to apply our theorems. We start with the second one: let $a,b \in G$ be arbitrary. For the natural projection $p_a$ and inclusion $i_b$, defined as in (\ref{proj_pt}) and (\ref{incl_pt}) we get that $Ui_b$ maps an element of $\mathcal H$ to $T_b{\mathcal D}:=\{\gamma\in \Gamma \vert \exists i \in {\mathcal D} \mbox{ such that } \gamma=T_b i\}$. Accordingly the value of $p_a U^*\phi$ depends only on the elements $\phi(\gamma)$ for $\gamma\in T_a {\mathcal D}$. Thus, if the distance between $T_a{\mathcal D}$ and $T_b{\mathcal D}$ is larger than $\rho$, the operator $p_a U^* A U i_b$ is equal to zero. As $\mathcal D$ is finite we can find a $R_{fr}>0$ such that $d_S(a,b)\geq R_{fr}$ implies $ p_a Hi_b=0$. Since the operator $A$ has the property (\ref{transinvA}) we get that $p_a Hi_b = p_{a+g} H i_{b+g}$ holds for all $a,b,g \in G$, which implies the $\mathcal C$-invariance.

\subsection{Percolation Hamiltonians}
In this section we discuss randomly coloured graphs and show the existence of frequencies $\nu_P$ for all $P\in \mathcal P$ along any F\o{}lner sequence. In fact we even give a precise formula to compute these frequencies.

We consider a randomly coloured infinite Cayley graph  $\Gamma=\Gamma(G,S)$, where each vertex is coloured in one of the $\vert\mathcal A\vert$ colours independently from the others. Each of the colours is chosen with the same probability, namely $\vert{\mathcal A} \vert^{-1}$. More precisely, we consider the following probability space: The sample space $\Omega =\left\{\omega=(\omega(x))_{x\in G}\ \vert\ \omega(x)\in {\mathcal A}\right\}$ is  the set of all possible colourings of the graph. We take $\mathcal B$ to be the $\sigma$-algebra of subsets of $\Omega$ generated by the finite-dimensional cylinders. Finally, we take the product measure $ \mu =\prod_{g\in G} \mu_g$ on $(\Omega, {\mathcal B})$, where $\mu_g$ is a probability measure on $\mathcal A$ given by
\[
 \mu_g(\omega(g)=a)=\frac{1}{\vert {\mathcal A}\vert}\hspace{1cm}\forall a\in {\mathcal A}.
\]
For $x\in G$ and $\omega \in \Omega $ a translation of a colouring $\omega$ is given by $\omega x : G\rightarrow {\mathcal A}$, $y\rightarrow \omega(yx^{-1})$. This defines a measure preserving action of $G$ on the measure space $(\Omega,{\mathcal B},\mu)$. Furthermore the independence of the randomly chosen colours implies ergodicity of the action of $G$ on $(\Omega,{\mathcal B},\mu)$. This means that if $A\in {\mathcal B}$ satisfies $Ax=A$ for all $x\in G$,  then $\mu(A)\in \{0,1\}$.

Let $P:D(P)\rightarrow \mathcal A$ be a pattern with domain $D(P)\in {\mathcal F}(G)$ containing the unit element. The set $A_P=\{ \omega \in \Omega\ \vert\ \omega\vert_{D(P)}=P \}$ consists of all colourings, which coincide with $P$ on $D(P)$. Let $f_P:\Omega\rightarrow \{0,1\}$ be the indicator function of $A_P$. Then we have
\begin{equation}\label{percineq}
 \sum\limits_{\gamma\in Q_{n,\diam(D(P))}}f_P(\omega\gamma^{-1})\leq\sharp_P({\omega}\vert_{ Q_n})\leq \sum\limits_{\gamma\in Q_{n}}f_P(\omega\gamma^{-1}).
\end{equation}
For a given F\o{}lner sequence $(Q_n)$ this immediately proves
\[
 \nu_P=\lim\limits_{n\rightarrow \infty}\frac{\sharp_P({\omega}\vert_{ Q_n})}{\vert Q_n \vert}
 \leq\lim\limits_{n\rightarrow \infty} \frac{1}{\vert Q_n\vert} \sum\limits_{\gamma\in Q_{n}}f_P(\omega\gamma^{-1}).
\]
One the other hand we have
\[
 \frac{1}{\vert Q_n\vert} \sum\limits_{\gamma\in \partial_{int}^R Q_n} f_P(\omega\gamma^{-1})\leq \frac{\vert\partial_{int}^R Q_n\vert}{\vert Q_n\vert}\rightarrow 0,\ n\rightarrow \infty
\]
for all $R>0$, which implies
\[
 \nu_P=\lim\limits_{n\rightarrow \infty}\frac{\sharp_P({\omega}\vert_{ Q_n})}{\vert Q_n \vert}
 \geq\lim\limits_{n\rightarrow \infty} \frac{1}{\vert Q_n\vert} \sum\limits_{\gamma\in Q_{n,\diam(D(P))}}f_P(\omega\gamma^{-1})
 =\lim\limits_{n\rightarrow \infty} \frac{1}{\vert Q_n\vert} \sum\limits_{\gamma\in Q_{n}}f_P(\omega\gamma^{-1})
\]
using the first inequality in (\ref{percineq}). If $(Q_n)$ is even a tempered F\o{}lner sequence, we can apply Lindenstrauss' pointwise ergodic theorem from \cite{Lindenstrauss-01}. Note that every F\o{}lner sequence has a tempered subsequence. Lindenstrauss' theorem proves that the above limit exists almost surely, in fact
\[
 \lim\limits_{n\rightarrow \infty} \frac{1}{\vert Q_n\vert} \sum\limits_{\gamma\in Q_{n}}f_P(\omega\gamma^{-1})= \EE (f_P)\hspace{1cm}\mbox{a.s.}
\]
Therefore for almost all configurations $\omega\in\Omega$ the frequency $\nu_P$ of an arbitrary pattern $P\in \mathcal P$ equals the expectation value $\EE (f_P)$. Thus we get
\[
\nu_P=\EE (f_P)=\mu(\omega \in A_P)= \vert {\mathcal A}\vert^{-\vert D(P)\vert} \hspace{1cm}\mbox{a.s.}
\]

\begin{Remark}\label{remark_pastur}
In the main theorem in section 4 we verified the uniform convergence of a sequence of distribution functions to a limit function.
The integrated density of states was then defined to be equal to this limit.
There is an alternative way to define the IDS in operator-algebraic terms. In fact it is in many cases possible to state a specific formula defining this distribution function. This relation is sometimes called Pastur-Shubin trace formula and applies to the mentioned models of (deterministic) periodic operators and (random) percolation Hamiltonians. In the deterministic case this formula reads
\[N(E)=\vert Q \vert^{-1}  \Tr  [ \chi_Q P((-\infty,E)) ]. \]
Here the set $Q$ is a fundamental domain of the group $G$ with respect to some finite index subgroup of $G$ (e.g. $G$ itself) and $P(I)$ is the spectral projection of the operator $H$ on an interval $I\subseteq \RR$. Note that the choice of the fundamental domain $Q$ does not influence the IDS. For a percolation Hamiltonian $H_\omega$ an expectation value enters the formula
\[N(E)=\vert Q \vert^{-1} \EE \left\{ \Tr  [ \chi_Q  P_\omega ((-\infty,E))]\right\}.\]
See \cite{LenzV-09} for a discussion of a general class of operators for which uniform convergence of the IDS and a Pastur-Shubin formula can be proven.
\end{Remark}

\subsection{Continuous dependence of the IDS on the Hamiltonian}

In the following we describe a continuity property of the IDS with respect to the change
of coefficients of the Hamiltonian. For this purpose we fix a graph $G$, a finite  alphabet ${\mathcal A}$,
a colouring ${\mathcal C}\colon G \to {\mathcal A}$, and sequences $(U_j)$ and  $(Q_n)$ as in Assumption \ref{ass:first}.
In particular we assume that the frequencies $\nu_P$ exist for all patterns $P\in \bigcup_{n\in \NN}{\mathcal P}(Q_n)$.
For a fixed $R\in \NN$ we consider the set of ${\mathcal C}$-invariant, selfadjoint operators on $\ell^2(G)$
with overall range $R$, and denote it by $L_{R, \mathcal C}$.

We will denote the trace class norm of a matrix $A$ by  $\|A\|_{\Tr} $.
Consider $\epsilon >0$ and two operators $H,G \in L_{R, \mathcal C}$ with the property
\[
 \forall \ x,y \in G\quad : \quad  |H(x,y) - G(x,y) | \le \epsilon
\]
and set $A=H-G$ respectively $A[U]= H[U]-G[U]$. By the decomposition $A[U] = A_{+}[U] -A_-[U]$, $A_{+}[U],A_{-}[U] \ge 0$
into the positive and negative part of $A[U]$ and
the triangle inequality of the trace class norm we have
\begin{align*}
 \|A[U]\|_{\Tr}
& \le \|A_{+}[U]\|_{\Tr} + \|A_{-}[U]\|_{\Tr}
\\
&\le \Tr A_{+}[U]  + \Tr A_{-}[U]
 \\
&= \sum _{x\in U} A_{+}[U](x,x)  + \sum _{x\in U} A_{-}[U](x,x)
 \\
&\le  \sum _{x\in U} \|A_{+}[U]\|  + \sum _{x\in U} \|A_{-}[U]\|
\\
&\le  |U| \,\|A_{+}[U]\|  + |U| \,\|A_{-}[U]\|.
\end{align*}
Since the positive and negative part of a selfadjoint operator
are obtained by the multiplication with an orthogonal projection, we obtain:
\[
\|A_{+}[U]\|  + \|A_{-}[U]\| \le  2\|H[U]-G[U] \|  .
\]
Now we estimate similarly as in Remark \ref{remark_bounded} using $\Vert A[U]\Vert \leq \Vert A\Vert$ and $\vert p_xAi_y\vert=\vert A(x,y)\vert \leq \epsilon $
\[
2\|H[U]-G[U] \|  \le  2   \vert B_R\vert \epsilon .
\]

For the spectral shift function of finite matrices
\[
 \xi(E,H[U], G[U]):= \Tr \left[  \chi_{(-\infty, E]}(H[U]) -   \chi_{(-\infty, E]}(G[U]) \right]
\]
it is well known that
\[
 \int_\RR dE \left | \xi(E,H[U], G[U])\right|
\le
\| H[U]-G[U]\|_{\Tr},
\]
 cf.~\cite{BirmanY-93}.
By the above bound we have
\[
 \int_\RR dE \left | \xi(E,H[U], G[U])\right|
\le 2\vert B_R\vert \, \epsilon \, |U| .
\]
Now for any $C^1$-function $f: \RR \to \RR $ with compact support
 we have
\[
\Tr  f((H[U]) ) - \Tr  f((G[U]) )  = \int_\RR  dE f'(E)  \,\xi(E,H[U], G[U])\  \le \Vert f'\Vert_\infty \int_\RR  dE   \,|\xi(E,H[U], G[U])| \ .
\]
For given $U\in \mathcal F(G)$ we set $N_H(f):=\int_\RR f(E)d N_H(E)$ and $N_H^U(f):=\int_\RR f(E)d N_H^U(E)=\Tr f(H[U])/|U|$, where $N_H^U(E):={n(H[U])}/{\vert U\vert}$ and consider the following difference
\begin{equation*}
 \left|N_H(f)-N_G(f)\right|
\le
\left|N_H(f)-N_H^U(f) \right| + \left|N_H^U(E)  -  N_G^U(E)\right|  + \left|N_G^U(E)  -  N_G(E)\right| .
\end{equation*}
We estimate the term $|N_H^U(f)-N_G^U(f)|$ independently of $U$ using the above calculations
\begin{equation*}
 |N_H^U(f)-N_G^U(f)|
=\frac{1}{|U|} |\Tr f((H[U]))- \Tr f((G[U]))|
\leq 2\Vert f'\Vert_\infty \vert B_R\vert \epsilon.
\end{equation*}
Let $ k\in \NN$ be such that the support of $f$ is contained in  the bounded $I=(-k,k)\subseteq \RR$ and consider
\[
 \left|N_H(f)-N_H^{U}(f) \right|
= \left| \int_I f(E)d N_H(E)- \int_I f(E) dN_H^{U}(E)    \right|
= \left| \int_I f(E)\, d\left(N_H(E)-N_H^{U}(E)\right)    \right|.
\]
We apply partial integration for Stieltjes integrals and obtain
\begin{align*}
  \int_\RR f(E)\, d\left(N_H(E)-N_H^{U}(E)\right)
&= f(E)\left(N_H(E)-N_H^{U}(E)\right)\Big\vert_{-k}^k - \int_I \left( N_H(E)-N_H^{U}(E) \right) d f(E) \\
&= \int_I \left(N_H^{U}(E)-N_H(E)\right) f'(E) d E .
\end{align*}
Now choose $U$ to be an element of the sequence $(U_{j,R})$ and obtain
\[
 \left|N_H(f)-N_H^{U_{j,R}}(f) \right|
\leq \Vert N_H-N_H^{U_{j,R}} \Vert_\infty \int_I \vert f'(E)\vert dE
\leq \delta(j,n) \Vert f' \Vert_\infty  \, 2k
\]
where
\[
\delta(j,n) =
8\frac{\vert \partial^RQ_n \vert}{\vert Q_n\vert}
+ (1+4 \vert B_R\vert )\frac{\vert \partial^{\diam(Q_n)}U_j  \vert}{\vert U_j\vert}
+\sum\limits_{P\in {\mathcal P}(Q_n)} \Big \vert \frac{\sharp_P ({\mathcal C}\vert_{U_j})}{\vert U_j\vert }- \nu_P  \Big \vert
+ \frac{\vert\partial_{int}^{R}U_j\vert}{ \vert U_j\vert }
\]
 denotes the upper bound in Theorem \ref{ids_unif} on the uniform convergence of the approximants.
As $\lim_{n\to\infty}\lim_{j\to\infty}\delta(j,n)=0$ we finally get
\[
 \left|N_H(f)-N_H^{U_{j,R}}(f) \right|  \rightarrow 0
\]
for $j \to \infty$, which proves
\[
 |N_H(f)-N_G(f)| \le 2 \Vert f' \Vert_\infty \vert B_R\vert \, \epsilon .
\]
for any compactly supported $f$ in $C^1$.

As a conclusion, we see that the IDS depends continuously (in the weak topology) on the variation of the coefficients in the matrix of the
Hamiltonian $H$.

\section{Open problems}
Here we discuss certain questions which arise naturally in connection with
our theorems.
\subsection*{Geometric properties of F\o lner sequences}
\begin{itemize}
 \item  The first question concerns the geometric requirement on the group $G$,
respectively the Cayley graph. It needed to contain a F\o lner sequence $(Q_n)$ with the additional property  that each set $Q_n$ symmetrically tiles the group, i.e. for each $n\in\NN$ one can find a set $K_n=K_n^{-1}\subseteq G$ such that $\{Q_n g\ \vert\ g\in K_n\}$ is a partition of $G$.
It is clear that this property holds if
there exists a sequence of subgroups $(G_n)$ such that one can choose the associated fundamental domains $(Q_n)$ as F\o lner sequence.

Thus one is led to ask whether the set of groups containing a F\o lner sequence of symmetrically tiling sets can be efficiently characterised.
From \cite{Weiss-01, Krieger-07} one can infer that all groups which are both amenable and residually finite satisfy this condition.
However a necessary condition is still missing.
For a discussion of such problems see for instance \cite{OrnsteinW-87, Weiss-01}.
 \item
Let us discuss some related properties of growing subsets of $G$.
As mentioned,  for any finite index subgroup $G'$ of $G$ each fundamental domain $F$ is a monotile of $G$.
One could ask if the converse holds true. More precisely, if $F$ tiles $G$: is there a finite index subgroup $G'$ of $G$?
Is $F$ even an fundamental domain of $G'$?
Another question is based on the extension of \cite{Weiss-01} due to Krieger \cite{Krieger-07} who showed: \emph{Assume the existence of a decreasing sequence $(H_n)_{n\in \NN}$ of finite index subgroups of $G$ such that $\bigcap_{n\in\NN}H_n=\{e\}$. Then there exists a sequence of subgroups $(G_n)_{n\in \NN}$ of $G$ and a sequence of associated fundamental domains $(Q_n)_{n\in \mathbb N}$ such that $(Q_n)$ is a F\o{}lner sequence.} Again, one can ask whether the converse holds. Given a group $G$, a sequence of subgroups $(G_n)$ and F\o lner sequence $(Q_n)$ such that each $Q_n$ is a fundamental domain of $G_n$ with respect $G$, does it follow that:
\begin{itemize}
\item $G_\infty := \bigcap_{n\in\NN} G_n$ is the unit element and
\item there exists an decreasing subsequence in $(G_n)$?
\end{itemize}
 \item The next open questions concerns the assumption that each element of the F\o lner sequence tiles the group with a symmetric grid. One could suppose that if there is some (not necessarily symmetric) grid $K\subseteq G$ along which $Q$ tiles $G$, one can find a set set $\tilde K=\tilde K^{-1}$ such that $Q$ tiles $G$ along $\tilde K$ as well.\newline
 Weiss proved in \cite{Weiss-01} that all groups that are built up by extensions using Abelian groups or finite groups obey F\o lner sequences of monotiles. This includes in particular the set of all solvable groups. Therefore all these groups whould fit in our setting as well if the above conjecture holds.
\end{itemize}

\subsection*{Approximation of the IDS with respect to different norms}
\begin{itemize}

\item The main result of section three is Theorem \ref{ergthm},
which is an ergodic type theorem for certain functions mapping finite subsets of the vertices of a graph $G$ into a Banach space $X$.
In the sequel we used this to obtain uniform estimates for the integrated density of states by setting $X=\mathcal B(\RR)$,
the Banach space of the right-continuous, bounded functions $f:\RR\rightarrow \RR$ equipped with supremum norm.\newline
One could think of other applications of Theorem \ref{ergthm} considering different Banach spaces, e.g. $X=L^p(\RR)$ for $p\in[1,\infty)$.
This would lead to estimates for the speed of convergence in terms of $L^p$-norms.
This seems to be a fruitfull approach when considering Schr\"odinger operators on $\RR^d$, cf.~\cite{GruberLV}.
More generally, our abstract approach would even allow to study convergence of counting functions which are elements of an appropriately chosen Fr\'echet space, cf.~Remark \ref{Frechet}.

\end{itemize}

\subsection*{Extension of the methods for other classes of operators}
\begin{itemize}
\item The operators which we considered in Section 4 were defined on $\ell^2(G,\mathcal H)$, where $\mathcal H$ was a finite dimensional Hilbert space. Physically, $\mathcal H$ encodes the internal degrees of freedom associated to a fixed vertex of $G$. There are interesting situations where the internal space is infinitely dimensional, for instance $\mathcal H=L^2(U)$, where $U$ is an open subset of $\RR^d$, or $\mathcal H=\ell^2(\ZZ)$.\newline
Our results cannot be directly applied to such models, since we used the finite dimension of $\mathcal  H$ at several steps in the proof. It would be interesting to implement a version of our methods which work also for infinite dimensional $\mathcal H$.

\item
A strong restriction on the type of operators which we can treat with the methods of this paper is that they have to be of finite hopping range.
The same restriction applies to  most of the papers  mentioned in the introduction, for instance \cite{MathaiSY-03,KlassertLS-03,Veselic-05b,LenzV-09}.
In particular it is not clear how the results of these papers can be extended to the case when the difference operator on $\ell^2(G)$
has a matrix kernel with very fast off-diagonal decay, albeit is not of finite range.

\item
In the context of \emph{stochastic} difference operators the methods mentioned in the previous item
may be indeed applicable even it the finite hopping range condition is violated.
This concerns operators such that the probability that a matrix element is non-zero decays to zero
as one moves away from the diagonal. An example of an operator falling into this class is
the adjacency operator associated to  a long range percolation graph, cf.~\cite{AntunovicV-09}.

\end{itemize}

\section{Appendix}
For the convenience of the reader we provide here certain rather standard technical results. One of them is the proof of Lemma \ref{lemma_abc} while the others are used for the verification of the almost additivity of the cumulative eigenvalue counting function in section four.

\begin{proof}[Proof of Lemma \ref{lemma_abc}] For any $Q\in {\mathcal F}(G)$ the equivalences
\begin{align*}
   g\in \bigcup\limits_{s\in B_R}(Q\setminus  Qs) \quad &\Leftrightarrow \quad \exists  s\in B_R : g\in Q, g\not\in Qs \\
 \Leftrightarrow\quad \exists  s\in B_R : g\in Q, gs^{-1}\not\in Q \quad &\Leftrightarrow \quad 1\leq d_S(g,G\setminus Q)\leq R  \quad \Leftrightarrow \quad g\in \partial_{int}^R(Q)
 \end{align*}
hold. This implies the first part of (a), the second part follows by a similar argument:
 \begin{equation*}
    g\in \hspace{-0.2cm}\bigcup\limits_{s\in B_R}(Qs\setminus  Q)
\quad \Leftrightarrow  \quad \exists  s\in B_R : g\not\in Q, gs^{-1}\in Q
\quad \Leftrightarrow  \quad  1\leq d_S(g, Q)\leq R
\quad \Leftrightarrow  \quad g\in \partial_{ext}^R(Q).
 \end{equation*}
 We turn to claim (b). Using (a) we obtain both
\[0\leq \frac{\vert \partial_{int}^R U_{j}\vert}{\vert U_j\vert}=\frac{\big\vert \bigcup\limits_{s\in B_R}(U_j\setminus  U_j s) \big\vert}{\vert U_j\vert}
\leq \frac{ \sum\limits_{s\in B_R}\vert U_j\setminus  U_j s\vert }{\vert U_j\vert}
=   \sum\limits_{s\in B_R}\frac{\vert U_j\setminus  U_j s\vert }{\vert U_j\vert}
=   \sum\limits_{s\in B_R}\frac{\vert U_j s\setminus  U_j\vert }{\vert U_j\vert}\rightarrow 0,\]
and
\[
 0\leq  \frac{\vert \partial^R_{ext} U_j \vert}{\vert U_j\vert}
= \frac{\vert \bigcup\limits_{s\in B_R}(U_js\setminus  U_j)\vert}{\vert U_j\vert}
\leq  \frac{ \sum\limits_{s\in B_R}\vert U_js\setminus  U_j \vert}{\vert U_j\vert}
= \sum\limits_{s\in B_R} \frac{ \vert U_js\setminus  U_j \vert}{\vert U_j\vert}\rightarrow 0
\]
for $j\rightarrow \infty$, as $(U_j)_{j\in \NN}$ is a F\o{}lner sequence. This verifies the first two limits in (b), the third limit follows.
Since the proofs of (c) and (d) are similar, we prove only (c). Let $(U_j)$ be a F\o{}lner sequence. We firstly show
 \begin{equation}\label{lemma_abc2}
  US\subseteq U^1\hspace{1cm}\mbox{ for all } U\in \mathcal F(G).
 \end{equation}
  Choose $x\in US$ arbitrary, then there exist $u\in U$ and $s\in S$ such that $us=x$. The case $x\in U$ is clear since $U\subseteq U^1$. If $x\notin U$ we get $d_S(u,x)=d_S(u,us)=1$ which implies $x\in U^1$ and (\ref{lemma_abc2}) follows.
  Now we claim
  \begin{equation}\label{lemma_abc3}
  \partial_{ext}^1 (U_{R})\subseteq \partial_{int}^{R} (U)\hspace{1cm}\mbox{ for all } U\in \mathcal F(G).
 \end{equation}
  For any $x\in\partial_{ext}^1 (U_{R})$ we have $d_S(x,U_R)=1$, which means that there exists $y\in U_R$ such that $d_S(x,y)=1$. Suppose that $x\notin U$ then we have $y\in \partial^1_{int} (U)$ since $y\in U$. This would imply that $y\notin U_R$. Thus we have $x\in U$ and $d_S(x,G\setminus U)\leq R$ since $x\notin U_R$ which implies (\ref{lemma_abc3}).
 Observe that by using (b) we get $\lim_{j\rightarrow \infty}\vert U_j\vert^{-1} \vert U_{j,R}\vert=1$.
 Hence there exists a constant $j_0\in \NN$ such that $\vert U_j\vert^{-1} \vert U_{j,R}\vert\geq \frac{1}{2}$ for all $j\geq j_0$. Now the above facts imply
 \[
0\leq\frac{\vert U_{j,R}S \setminus U_{j,R}\vert}{\vert U_{j,R}\vert}
\stackrel{(\ref{lemma_abc2})}{\leq} \frac{\vert (U_{j,R})^1\setminus U_{j,R}\vert}{\vert U_{j,R}\vert}
=    \frac{\vert \partial^1_{ext}(U_{j,R})\vert}{\vert U_{j,R}\vert}
\stackrel{(\ref{lemma_abc3})}{\leq}  \frac{\vert \partial^{R}_{int}(U_{j})\vert}{\vert U_{j,R}\vert}
\leq 2\frac{\vert \partial^{R}_{int}(U_{j})\vert}{ \vert U_{j}\vert}
\]
for all $j\geq j_0$. Since $(U_j)_{j\in \mathbb N}$ is a F\o{}lner sequence applying (b) finishes the proof.
\end{proof}

\begin{Proposition}\label{prop1}
 Let $A$ and $C$ be selfadjoint operators in a finite dimensional Hilbert space, then we have
\[\vert n(A)(E)-n(A+C)(E) \vert \leq \rank(C)\]
for all $E\in\mathbb R$.
\end{Proposition}
Before starting the proof of the proposition we mention some general facts concerning selfadjoint operators on finite dimensional Hilbert spaces. Given a Hilbert space of dimension $n<\infty$ with a fixed basis the above operators are hermitian matrices of dimension $n\times n$. Since $n(A)$ is defined as the eigenvalue counting function, we are interested in the relation between the size of the eigenvalues of $A$ and $A+C$. This will be given with the help of the minmax principle of Courant and Fischer: Let $\lambda_1(A)\leq\dots\leq \lambda_n(A)$ be the eigenvalues of a hermitian matrix $A$, then the following equations hold
\begin{equation}\label{minmax1}
\lambda_k(A)=\min\limits_{\psi_1,\cdots,\psi_{n-k}\in \mathbb C^n}\max\limits_{  \underset {\Vert \phi \Vert=1} {\phi \perp \psi_1,\cdots,\psi_{n-k}}  }\langle \phi, A\phi\rangle
\end{equation}
and
\begin{equation}\label{minmax2}
\lambda_k(A)=\max\limits_{\psi_1,\cdots,\psi_{k-1}\in \mathbb C^n}\min\limits_{\underset{\Vert \phi \Vert=1}{\phi \perp \psi_1,\cdots,\psi_{k-1}}   }\langle \phi, A\phi\rangle
\end{equation}
For the proof of this see for example \cite{Horn-90}, where it is also stressed that the minimising respectively maximising vectors are exactly the eigenvectors: for all $k=1,\dots n$ let $f_k$ be the a normalised eigenvector for the eigenvalue $\lambda_k$, then we get
\begin{equation}\label{minmax3}
\lambda_k(A) =\min\limits_{\underset {\Vert \phi \Vert=1}{\phi \perp f_1,\cdots,f_{k-1}}}\langle \phi, A\phi\rangle    = \max\limits_{\underset{\Vert \phi \Vert=1}{\phi \perp f_{k+1},\cdots,f_{n}}}\langle \phi, A\phi\rangle
\end{equation}
\begin{proof}[Proof of Proposition \ref{prop1}]
For a given vector $\xi\in \CC^n$ the matrix $B:=  \vert \xi\rangle \langle \xi \vert:=\xi \xi^*$ is hermitian and of rank one. Given an arbitrary element $s\in \CC$ the equality $\langle \phi,sB\phi \rangle= s\langle \phi,\xi\rangle\langle \xi,\phi \rangle= s\vert\langle \phi,\xi\rangle\vert^2$ holds for all $\phi \in \mathbb C^n$. By using (\ref{minmax2}) we get
\begin{eqnarray*}
 \lambda_k(A+sB)&=&\max\limits_{\psi_1,\cdots,\psi_{k-1}\in \mathbb C^n}\min\limits_{\underset{\Vert \phi \Vert=1}{\phi \perp \psi_1,\cdots,\psi_{k-1}}}\langle \phi, A\phi\rangle+s\langle \phi, B\phi\rangle\\
&=&\max\limits_{\psi_1,\cdots,\psi_{k-1}\in \mathbb C^n}\min\limits_{\underset{\Vert \phi \Vert=1}{\phi \perp \psi_1,\cdots,\psi_{k-1}}}\langle \phi, A\phi\rangle+s \vert\langle\phi,\xi\rangle\vert ^2.
\end{eqnarray*}
We estimate the maximum from below by setting $\psi_i=f_i$, $i=1,\dots, k-2$ and $\psi_{k-1}=\xi$, where $f_i, i=1,\dots,k-1$ are the normalised eigenvectors
\[\lambda_k(A+sB)\geq\min\limits_{\underset{\Vert \phi \Vert=1}{\phi \perp f_1,\cdots,f_{k-2},\xi}}\langle \phi, A\phi\rangle+s\vert\langle\phi,\xi\rangle\vert^2
\geq \min\limits_{\underset{\Vert \phi \Vert=1}{\phi \perp f_1,\cdots,f_{k-2}} }\langle \phi, A\phi\rangle
\stackrel{(\ref{minmax3})}{=}\lambda_{k-1}(A).\]
Similarly, by using the relation (\ref{minmax1}) we get an upper bound for the $k$-th  eigenvalue of $A+sB$
\[
\lambda_k(A+sB)=\min\limits_{\psi_1,\cdots,\psi_{n-k}\in \mathbb C^n}\max\limits_{\underset{\Vert \phi \Vert=1}{\phi \perp \psi_1,\cdots,\psi_{n-k}}}\langle \phi, A\phi\rangle+s\vert\langle\phi,\xi\rangle\vert^2
\leq\max\limits_{\underset{ \Vert \phi \Vert=1}{\phi \perp f_{k+2},\cdots,f_{n},\xi} }\langle \phi, A\phi\rangle
\leq\lambda_{k+1}(A).
\]
This gives for arbitrary $s\in \CC$ and $\xi\in \CC^n$, that for the matrix $B:=\vert \xi\rangle \langle \xi \vert$ the inequalities
\begin{equation}\label{rank1}
        \lambda_{k-1}(A)\leq\lambda_k(A+sB)\leq\lambda_{k+1}(A).
\end{equation}
hold. Now let $C$ be a hermitian matrix of rank $m$. With the eigendecomposition we get $C=\sum_{i=1}^n s_i \xi_i \xi_i^*$, where $s_i$ are the eigenvalues and $\xi_i$ the normalised eigenvectors for $i=1,\dots,n$. Since the rank of $C$ equals $m$, there are exactly $m$ non-zero eigenvalues in the spectrum. Thus the above sum consists of $m$ summands. Induction over the rank of $C$, leads to
\begin{equation}\label{eig_ineq}
 \lambda_{k-m}(A)\leq \lambda_k\Big(A+\sum\limits_{i=1}^m t_iB_i\Big)\leq \lambda_{k+m}(A),
\end{equation}
where for each $i=1,\dots,m$ the matrix $B_i$ is of the form $B_i = \vert \xi_i\rangle \langle \xi_i \vert$ for some eigenvector $\xi_i$ and $t_i$ is the associated eigenvalue. From (\ref{rank1}) we know that this is true for $m=1$. The inductive step follows with
\[\lambda_{k-m-1}(A)\leq\lambda_{k-1}\Big(A+\sum\limits_{i=1}^m t_iB_i\Big)\leq \lambda_k\Big(A+\sum\limits_{i=1}^{m+1} t_iB_i\Big)\leq \lambda_{k+1}\Big(A+\sum\limits_{i=1}^{m} t_iB_i\Big)\leq \lambda_{k+m+1}(A).\]
Finally we will translate (\ref{eig_ineq}) into the language of the eigenvalue counting functions. Setting $k:=n(A+C)(E)$, leads to $\lambda_k(A+C)\leq E < \lambda_{k+1}(A+C)$, which implies
\[\lambda_{k-m}(A)\leq E < \lambda_{k+m+1}(A)\]
by using  (\ref{eig_ineq}). This gives us $k-m\leq n(A)(E)\leq k+m$ which proves
\begin{equation*}
 n(A)(E)-m \leq n(A+C)(E)\leq n(A)(E) + m
\end{equation*}
for all $E\in \mathbb R$ and the proposition follows.

\end{proof}
Applying this proposition on operators defined on a finite dimensional Hilbert space and their projections on a subspace, leads to the following proposition, which has already been proven in \cite{LenzS-06}.

\begin{Proposition}\label{prop2}
 Let $V$ be a finite dimensional Hilbert space and $U$ a subspace of $V$. If $i:U\rightarrow V$ is the inclusion and $p:V\rightarrow U$ the orthogonal projection, we have
\[\vert n(A)(E)-n(pAi)(E)\vert\leq 4 \cdot \rank (1-i\circ p)\]
for all selfadjoint operators $A$ on $V$ and all energies $E\in \mathbb R$.
\end{Proposition}
\begin{proof}
We set $P:=i\circ p: V\rightarrow V$ and use the triangle inequality to obtain
\begin{equation} \label{prop2_1}
\vert n(A)(E)-n(pAi)(E)\vert\leq \vert n(A)(E)-n(PAP)(E)\vert  +\vert n(PAP)(E)  -n(pAi)(E)\vert.
\end{equation}
With the help of the equality
\[A-PAP = (1-P)AP+PA(1-P)+(1-P)A(1-P)\]
and the previous proposition we get
\begin{eqnarray}
\vert n(A)(E)-n(PAP)(E)\vert &\leq&  \rank(PAP-A)\nonumber\\
&=& \rank((1-P)AP+PA(1-P)+(1-P)A(1-P))\label{prop2_2}\\
&\leq& 3\rank (1-P).\nonumber
\end{eqnarray}
Let $U^\perp$ denote the orthogonal complement of $U$ and define $0_{U^\perp}:U^\perp\rightarrow U^\perp$ with $f\mapsto 0$. It is obvious that
\[PAP=i\circ p\circ A\circ i\circ p = p\circ A\circ i \oplus 0_{U^\perp}\]
holds and therefore we have
\[\vert n(PAP)(E)  -n(pAi)(E)\vert=\vert n(p\circ A\circ i \oplus 0_{U^\perp})(E)  -n(pAi)(E)\vert=\vert n(0_{U^\perp})(E) \vert\leq \dim(U^\perp).\]
Note that the dimension of $U^\perp$ equals the rank of $(1-P)$. Using this and (\ref{prop2_2}) to estimate (\ref{prop2_1}) finishes the proof.
\end{proof}

The next proposition is a well known dimension argument from linear algebra which we will state for completeness reasons.

\begin{Proposition}\label{prop_help}
 Let $H$ be a finite dimensional Hilbert space, $U$, $V$ subspaces of $H$ with $\dim U > \dim V$ then $\dim V^\perp \cap U > 0$.
\end{Proposition}
The following result can be found in \cite{Simon-87} and \cite{LenzMV-08}. It is a useful tool in the proof of the Corollary \ref{cor_disc}, where the points of discontinuity of the IDS are characterised.

\begin{Proposition}\label{prop4}
 Let $A$ be a selfadjoint operator on a finite-dimensional Hilbert space $V$. Let $\lambda\in \mathbb R$ and $\epsilon>0$ be given and denote by $U$ the subspace of $V$ spanned by the eigenvectors of $A$ belonging to the eigenvalues in the open interval $(\lambda-\epsilon, \lambda+\epsilon)$. If there exist $k$ pairwise orthogonal and normalised vectors $u_1,\dots,u_k\in V$ such that $(A-\lambda)u_j$, $j=1,\dots,k$ are pairwise orthogonal and satisfy $\Vert (A-\lambda)u_j \Vert<\epsilon$, then $\dim(U)\geq k$.
\end{Proposition}
\begin{proof}
We assume $\dim (U)<k$. Let $S$ be the linear span of $u_1,\dots, u_k$. By Proposition \ref{prop_help} there exists an unit element $s\in S$, which is orthogonal to $U$, e.g $s\in U^\perp$. Hence, $s$ is a combination of elements $\overline u_k$ with $\Vert (A-\lambda)\overline u_k\Vert\geq \epsilon$. This gives us $\Vert (A-\lambda)s\Vert\geq \epsilon$. On the other hand we know that $s\in S$ is an unit element combined by elements $u_j$ with $\Vert(A-\lambda)u_j \Vert<\epsilon$, $j=1,\dots ,k$, which implies $\Vert(A-\lambda)s \Vert<\epsilon$.
\end{proof}

\subsection*{Acknowledgement} 
We thank M.~Keller and F.~Pogorzelski who pointed out an error in an earlier version of this paper.

\newcommand{\etalchar}[1]{$^{#1}$}
\def\cprime{$'$}\def\polhk#1{\setbox0=\hbox{#1}{\ooalign{\hidewidth
\lower1.5ex\hbox{`}\hidewidth\crcr\unhbox0}}}
%

\end{document}